\documentclass[11pt]{article}
\usepackage[margin=1in]{geometry}

\usepackage{amsfonts,amsmath,dsfont,mathtools}
\usepackage{amssymb} 
\usepackage[hypertexnames=false]{hyperref} 
\usepackage{amsthm,comment}
\usepackage[inline, shortlabels]{enumitem}
\usepackage{algorithm2e}
\usepackage{xcolor}

\hypersetup{ 
    colorlinks=true,
    citecolor=red,
    linkcolor=red,
    urlcolor=red
}

\usepackage[textsize=tiny,disable]{todonotes}
\newcommand{\gtodo}[1]{\todo[color=yellow]{#1}}

\newcommand{\itodo}[1]{\todo[color=green]{#1}}

\usepackage[noabbrev,capitalise,nameinlink]{cleveref}
\crefname{enumi}{}{}
\crefname{equation}{}{}
\crefname{claim}{Claim}{Claims}

\newtheorem{theorem}{Theorem} 
\newtheorem{lemma}[theorem]{Lemma}

\theoremstyle{remark}

\theoremstyle{definition}
\newtheorem{definition}[theorem]{Definition}
\theoremstyle{remark}
\newtheorem{remark}[theorem]{Remark}

\newcommand{\Exp}{\ensuremath{\operatorname{\mathbb{E}}}}
\renewcommand{\Pr}{\ensuremath{\operatorname{\mathbb{P}}}}
\renewcommand{\AA}{\mathcal{A}}
\newcommand{\EE}{\mathcal{E}}
\newcommand{\BB}{\mathcal{B}}
\newcommand{\GG}{\mathit{\Gamma}}
\newcommand{\HH}{\mathcal{H}}
\newcommand{\WW}{\mathcal{W}}
\newcommand{\XX}{\mathcal{X}}
\DeclareMathOperator{\poly}{poly}
\DeclareMathOperator{\diam}{diam}
\newcommand{\ErdosRenyi}{Erdős-Rényi}
\newcommand{\twostate}{2-state MIS process}

\newcommand{\black}{\mathtt{black}}
\newcommand{\white}{\mathtt{white}}
\newcommand{\gray}{\mathtt{gray}}
\newcommand{\on}{\mathtt{on}}
\newcommand{\off}{\mathtt{off}}
\newcommand{\level}{\mathit{level}}
\newcommand{\blkone}{\mathtt{black1}}
\newcommand{\blkzero}{\mathtt{black0}}

\title{Distributed Self-Stabilizing MIS with Few States and Weak Communication}

\author{George Giakkoupis\\
Inria, Rennes, France\\
\texttt{george.giakkoupis@inria.fr}
\and
Isabella Ziccardi\\
Bocconi University, Milan, Italy\\
\texttt{isabella.ziccardi@unibocconi.it}
}

\date{}

\begin{document} 

\maketitle 

\begin{abstract}
We study a simple random process that computes a maximal independent set (MIS) on a general $n$-vertex graph.
Each vertex has a binary state, black or white, where black indicates inclusion into the MIS.
The vertex states are arbitrary initially, and  are updated in parallel: 
In each round, every vertex whose state is ``inconsistent'' with its neighbors', i.e., it is black and has a black neighbor, or it is white and all neighbors are white, changes its state with probability $1/2$.
The process stabilizes with probability 1  on any graph, and the resulting set of black vertices is an MIS.  
It is also easy to see that the expected stabilization time is $O(\log n)$ on certain graph families, such as cliques and trees.
However, analyzing the process on graphs beyond these simple cases seems challenging.

Our main result is that the process stabilizes in $\poly(\log n)$ rounds w.h.p.\ on $G_{n,p}$ random graphs, for $0\leq p \leq  \poly(\log n)\cdot n^{-1/2}$ and $p \geq 1/\poly(\log n)$.
Further, an extension of this process, with larger but still constant vertex state space, stabilizes in $\poly(\log n)$ rounds on $G_{n,p}$ w.h.p., for all $1\leq p\leq 1$.
We conjecture that this improved bound holds for the original process as well.
In fact, we believe that the original process stabilizes in $\poly(\log n)$ rounds \emph{on any given $n$-vertex graph} w.h.p.
Both processes readily translate into distributed/parallel MIS algorithms, which are self-stabilizing, use constant space (and constant random bits per round), and assume restricted communication as in the beeping or the synchronous stone age models.
To the best of our knowledge, no previously known MIS algorithm is self-stabilizing, uses constant space and constant randomness, and stabilizes in $\poly(\log n)$ rounds in general or random graphs.
\end{abstract}


\section{Introduction}

Finding a maximal independent set (MIS) is a fundamental problem in parallel and distributed computing.
Given a graph $G=(V,E)$, the objective is to identify a set of vertices $S\subseteq V$ such that no two vertices $u,v\in S$ are adjacent to each other (\emph{independence} property), and no vertex $u\in V\setminus S$ can be added to $S$ without violating independence (\emph{maximality} property).
The significance of the problem in parallel computing was first recognised in the early 80s~\cite{Valiant82,Cook83}, due to its various applications in symmetry breaking~\cite{Luby86}, and it has been studied extensively every since (see \cite{BarenboimEPS16} for a review of work until 2015, and \cite{BalliuBHORS21,GhaffariGR21} for state of the art results).




In this paper we explore simple distributed random processes on graphs that find an MIS starting from arbitrary initial states of the vertices.
These processes immediately translate into self-stabilizing~\cite{Dijkstra74,Dolev2000} synchronous distributed algorithms for network systems with severely restricted computation and communication capabilities, such as wireless sensor networks.
The processes we consider are also relevant to certain biological cellular networks.
For example, it is known that a biological process occurring during the development of the nervous system of a fly is equivalent to computing an MIS~\cite{AfekABHBB11,JeavonsS016}.

The main random process we consider, which we call the \emph{2-state MIS process},
is as follows.
Each vertex has a binary state, \emph{black} or \emph{white}, where black indicates inclusion into the MIS.
The vertex states are arbitrary initially and are updated in synchronous rounds. 
In each round, every vertex $u$ whose state violates the independence or maximality properties, i.e., $u$ is black and has a black neighbor, or it is white and has no black neighbor, changes its state to the opposite state with probability $1/2$.
It is easy to see that the state of a vertex stabilizes as soon as it is black and has no black neighbors, or it is white and has a \emph{stabilized black} neighbor; and when all vertices have stabilized, the set of black vertices is an MIS.
It is also immediate that, on any graph $G$, the process stabilizes eventually with probability 1 (due to the randomization)
.\footnote{We could have defined the process so that the transition from white to black (when the white vertex has no black neighbors) occurs with probability $1$, but we opted for a randomized transition because it simplifies our analysis.}

The 2-state MIS process can be viewed as a natural parallelization (with the addition of randomness) of a simple self-stabilizing sequential deterministic algorithm, proposed in~\cite{Shukla1995,Hedetniemi2003}, where in each step a single node updates its state (from black to white, if the node has a black neighbor, and from white to black if it has no black neighbors).
\cite{Shukla1995} also observed that by randomizing the transitions of the sequential algorithm we obtain an algorithm that stabilizes with probability 1 on a general adversarial scheduler model, which includes the synchronous model.
A similar observation follows from a general transformation framework proposed in~\cite{TurauW06}.
The sequential algorithm is know to stabilize after each process has taken at most 2 steps (regardless of the scheduling order).
However, analyzing the stabilization time of the parallel process seems a much more challenging problem, and has not been studied until now.

The 2-state MIS process directly translates into a self-stabilizing MIS algorithm for the harsh beeping communication model~\cite{CornejoK10}.
In that model, in every synchronous round,
each node either listens or beeps, and a listening node can only differentiate between none of its neighbors' beeping, or at least one beeping.
In our case, we can let black nodes beep in each round, while white nodes listen. 
Black nodes must be able to detect collisions (otherwise they cannot tell if they have a black neighbor), thus we assume the beeping model version with \emph{sender collision detection} (a.k.a.\ full-duplex model)~\cite{AfekABCHK13,Ghaffari17}.

We also propose 
a simple variant of the 2-state MIS process, called the \emph{3-state MIS process},
which has an additional state and does not require collision detection (see \cref{def:3-state-MIS}).
This variant is suitable for the synchronous  stone age model~\cite{EmekW13,EmekK21}. 
The synchronous  stone age model can be viewed as an extension of the beeping model over a constant number of channels (without collision detection): each node beeps in at most one channel and listens to the other channels. 


Overall, the algorithms obtained from the 2-state and 3-state MIS processes have several attractive properties: they use a constant number of states (2 or 3) and one random bit per round, they do not require node IDs or any global graph information (such as the number of vertices $n$ or the maximum degree $\Delta$), assume very week communication (the beeping or stone age models), they are self-stabilizing, and are extremely simple.
We will prove that, on some families of graphs, these algorithms are also \emph{fast}, i.e., they stabilize (from an arbitrary initial state) in a number of rounds that is poly-logarithmic in $n$, w.h.p.\footnote{In this paper, we do not analyze the 3-state MIS process, but we expect that it behaves similarly (or better than) the 2-state MIS process.}
Moreover, despite that we were not  able to prove such as strong result here, we believe that these algorithms are fast \emph{in all graphs}.

Several self-stabilizing distributed MIS algorithms have been proposed in the literature, but as far as we know, none possesses all the above properties. 
Known self-stabilizing MIS algorithms for the beeping model require (approximate) knowledge of $n$, use space that is a super-constant function of $n$, and require a super-constant number of random bits \cite{AfekABCHK13,JeavonsS016,Ghaffari17}.
In the stone age model, an MIS algorithm proposed in~\cite{EmekW13} has similar properties as our algorithms (and is provably fast for all graphs) but is not self-stabilizing;
while a self-stabilizing algorithm for the model proposed recently in~\cite{EmekK21} is fast only on graphs whose diameter is bounded by a known constant $D$.
Other randomized self-stabilizing MIS algorithms required super constant state and communication~\cite{Turau19}.
Finally, known deterministic self-stabilizing MIS algorithms require distinct node IDS, super constant state and communication, and are in general much slower than the randomized algorithms, stabilizing in time linear in $n$ or in the maximum degree $\Delta$~\cite{Ikeda2002,GoddardHJS03,Turau07,BarenboimEG22}. 


\subsection{Our Contribution}

We first analyze the stabilization time of the 2-state MIS process on complete graphs and on graphs with bounded arboricity.\footnote{The arboricity of a graph is the minimum number of forests into which we can partition its edges.}
We also provide an upper bound in terms of the maximum degree for general graphs.
The proof of these results is mostly straightforward.

\begin{theorem}
    \label{thm:intro-simple}
    The stabilization time of the 2-state MIS process on $n$-vertex graph $G$ is
    \begin{itemize}
        \item 
          $O(\log n)$ in expectation and $O(\log^2 n)$
         \itodo{$\Theta(\log^2 n)$?} 
          w.h.p., if $G$ is the complete graph $K_n$. 
        
        \item
        $O(\log n)$ w.h.p., if $G$ has bounded arboricity.

        \item
        at most $O(\Delta\log n)$ w.h.p., if the maximum degree of $G$ is $\Delta$.     
    \end{itemize}
\end{theorem}

A main technical contribution of the paper is the analysis of the 2-state MIS process on \ErdosRenyi\ $G_{n,p}$ random graphs.
We show a poly-logarithmic upper bound 
for $G_{n,p}$ random  graphs when the average degree $np$ is at most $\poly(\log n)\cdot \sqrt n$. 
The same bound is easily obtained also when the average degree is at least $n/\poly(\log n)$.

\begin{theorem}
    \label{thm:intro-2s}
    The stabilization time of the 2-state MIS process on a $G_{n,p}$ random graph, such that $0\leq p\leq \poly(\log n)\cdot n^{-1/2}$ or $p\geq 1/\poly(\log n)$, is at most $\poly(\log n)$ w.h.p.
\end{theorem}

Our proof techniques do not yield a poly-logarithmic upper bound for the 2-state MIS process on $G_{n,p}$ for the complete range of $p$.
Our second technical contribution is an extension of the 2-state MIS process that provably stabilizes in poly-logarithmic time w.h.p.\ on $G_{n,p}$ for all $0\leq p\leq 1$.
The extended process uses a phase clock sub-process proposed in~\cite{EmekK21}.
Interestingly, unlike~\cite{EmekK21}, we do not use the phase clock for synchronization, but rather as a local non-synchronized counter (see \cref{sec:Techniques} for a more detailed discussion).

\begin{theorem}
    \label{thm:intro-3c}
    There is an extension of the $2$-state process, with $18$ states, such that 
    the stabilization time of the process on a $G_{n,p}$ random graph, for any $0\leq p\leq 1$,  is at most $\poly(\log n)$ w.h.p.
\end{theorem}

We believe that the bound of \cref{thm:intro-3c} holds for the 2-state MIS process, as well.
In fact, we conjecture that the stabilization time of the  2-state MIS process is $\poly(\log n)$ w.h.p.\ on \emph{any given $n$-vertex graph}.
We also conjecture that the same is true for the 3-state MIS process.
For the 2-state process, the best general upper bound we can hope for is $O(\log^2 n)$, as the process requires $\Theta(\log^2 n)$ rounds to stabilize on the complete graph $K_n$ w.h.p.\footnote{
It also requires $\Theta(\log^2 n)$ rounds \emph{in expectation} to stabilize on a graph consisting of $\sqrt{n}$ disjoint cliques $K_{\sqrt n}$.}
For the 3-state process, we have no example of a graph where the stabilization time is larger than $O(\log n)$.

\subsection{Analysis Overview and Techniques}
\label{sec:Techniques}

Below we  give an overview of the analysis of the 2-state MIS process and its extension, on $G_{n,p}$ random graphs.

To avoid having to deal simultaneously with the randomness of the graph and the arbitrary initialization of vertex states, we deal with graph randomness first.
We define a family of \emph{good graphs}, containing those graphs that satisfy all structural properties that we will need for the analysis, e.g., bounds on the average degree of any induced subgraph, and bounds on the number of common neighbors of any two vertices 
(see \cref{def:good-graph}).
We then show that a $G_{n,p}$ random graph is good  w.h.p., and assume an arbitrary good graph in the analysis.

The analysis proceeds by showing that starting from any vertex states, the process makes sufficient progress after $O(\log n)$ rounds, where progress is measured by the expected number of vertices that stabilize.

In the 2-state MIS process, 
we call a vertex \emph{active} if it is black and has a black neighbor, or it is white and has no black neighbors.
Thus, active vertices change their state to a uniformly random state in the next step. 
A vertex is  \emph{$k$-active} if it is active and has \emph{at most} $k$ active neighbors.

An elementary property of the 2-state MIS process is that if a vertex is $k$-active, then it becomes stabilized black in $O(\log k)$ rounds with probability $\Omega(1/k)$.
We also use an extension of this property to sets of active vertices.\footnote{Similar properties are commonly used in the analysis of distributed MIS algorithms in the literature.}
%
%
These two properties, combined with structural properties of good graphs, suffice to show the desired expected progress in the case in which the number of non-stabilized vertices or the number of active vertices is large enough. 

The more difficult case is when the number of non-stabilized vertices is relatively small, namely $O(p^{-1}\log^2n)$, and a smaller than $1/\poly(\log n)$ fraction of them are active.
One may expect this to be an easy case, since the induced subgraph on a \emph{random} subset of $O(p^{-1}\log^2n)$ vertices has maximum degree $\Delta = O(\log^2n)$ w.h.p.\ (and \cref{thm:intro-simple} gives an $O(\log^3n)$ bound for that $\Delta$).
However, the above bound on $\Delta$ does not apply to an induced subgraph on an \emph{arbitrary} subset of $O(p^{-1}\log^2n)$ vertices.
Nevertheless, it is true that the \emph{average} degree is $O(\log^2n)$, thus a constant fraction of vertices have degree $O(\log^2n)$.

Let $u$ be one such vertex, i.e., of degree $d = O(\log^2n)$ in the induced subgraph of non-stabilized vertices.
To prevent $u$ from becoming active (and thus $d$-active) or becoming stabilized, in each round at least one neighbor of $u$ must be non-stabilized black.
We show that, roughly speaking, if a vertex $v$ has probability $b$ of being non-stabilized black  at some point during an interval of $r$ rounds (ignoring the first few rounds, e.g., if $v$ is black initially) then $v$ has probability $\poly(b/r)$ of becoming $\theta$-active in that interval.
For the purposes of the analysis, it suffices to set $r = O(\log \log n )$.
Then $\theta$ is, roughly, bounded by the maximum number of common neighbors two nodes may have, thus $\theta \leq \poly(\log n)$ if $p \leq \poly(\log n)\cdot n^{-1/2}$
(see \cref{sec:refined} for the relevant lemmas).

If each of the $d$ neighbors of $u$ has probability less than $1/(2d)$ of becoming non-stabilized black in the next $r$ rounds, 
then $u$ has a constant probability of becoming \itodo{What if we write d-active? (and, therefore, it stabilizes with prob. 1/d)}
active (or stabilize).
On the other hand, if there is some neighbor $v$ that has probability $b \geq 1/(2d)$ of becoming non-stabilized black in the next $r$ rounds, we saw above that $v$ becomes $\theta$-active with probability  at least $\poly(b/r) = 1/\poly(\log n)$.
We conclude that, with probability $1/\poly(\log n)$, $u$ is $\poly(\log n)$-active or has some $\poly(\log n)$-active neighbor at some point in the next $r = O(\log\log n)$ rounds.
It follows that $u$ stabilizes with probability $1/\poly(\log n)$ in the next $O(\log n)$ rounds.\footnote{We suspect that a refinement of this argument may be useful for a broader class of graphs.}

When $p > \poly(\log n) \cdot n^{-1/2}$, the last case of the analysis above does not give a poly-logarithmic bound.
A way to overcome this problem is to control how often a vertex can change its state from white to black.
We extend the 2-state MIS process by incorporating such a control mechanism.

We call the new process the \emph{3-color MIS process}.
It consists of two sub-processes running in parallel: 
The first is similar to the 2-state MIS process with the addition of a third color, \emph{grey};
a black vertex now becomes gray instead of white, a gray vertex becomes white after a while, and other vertices treat gray vertices as white.
The transition from gray to white is controlled by the second sub-process, called the \emph{logarithmic switch}.

In the logarithmic switch, each vertex has an on/off binary variable, and a gray vertex changes to white if the switch variable of the vertex is on.
We would like that the logarithmic switch satisfy two basic properties: 
(i)~a vertex switches from off to on every $\Theta(\log n)$ rounds;
and (ii)~it switches from on to off every $O(1)$ rounds.\footnote{The reason why a logarithmic switch suffices, rather than a `double-logarithmic' switch is that, in the induced subgraph on  $O(p^{-1}\log^2n)$ vertices consider in the last case of the analysis of the 2-state MIS process, a constant fraction of vertices have at most $O(\log n)$ neighbors of degree $\Omega(\log^3 n)$.}
However, we do not know how to implement property~(i) using constant states.
We observe that it suffices if property~(i) is satisfied only when $p > \poly(\log n)\cdot n^{-1/2}$; for smaller $p$, a weaker property suffices:
(i$'$)~a vertex switches from off to on after \emph{at most} $O(\log n)$ rounds.
It is not immediately obvious how to implement this distinction, because we want the process to work for all $0\leq p\leq 1$ \emph{without knowing $p$} (or anything else about the graph topology).
We achieve that as follows.

We exploit the fact that if $p > \poly(\log n)\cdot n^{-1/2}$ then the graph has constant diameter (in fact diameter 2).
The logarithmic switch process we devise is similar to the phase clock process \emph{RandPhase} proposed in \cite{EmekK21}.
RandPhase assumes that an upper bound $D$ on the graph diameter is available to the process and uses $D+3$ states.
The core mechanism of 
the logarithmic switch is identical to that of RandPhase for $D = 3$ (not 2!), but the underlying graph may have arbitrary (and unknown) diameter.
The logarithmic switch includes also a mapping of the states to the on/off values of the switch.
Unlike RandPhase which is used for sychronization (it achieve synchronous phases of length $D+\Theta(\log n)$), the purpose of the logarithmic switch is not synchronization, as it is not required that the switch variables of different vertices change simultaneously.

\paragraph{Roadmap.}
The rest of the paper is organized as follows.  
\cref{sec:2-state-3-state-def} contains the definition and some basic properties of the 2-state and 3-state MIS processes. 
\cref{sec:simple-bounds} provides a proof of \cref{thm:intro-simple}. 
\cref{sec:2-state-random-graphs} proves \cref{thm:intro-2s}. 
\cref{sec:log-switch-3-color} defines the 3-color MIS process and proves \cref{thm:intro-3c}.
And \cref{sec:related-work} reviews related work. 
\gtodo{we must update that in the end}

\paragraph{Notation.}
\label{sec:2-state-mis}

Let $G = (V,E)$ be a graph on $n$ vertices.
For each vertex $u\in V$, $N(u) = \{v\colon (u,v)\in E\}$ is the set of neighbors of $u$, and $N^+(u) = N(u)\cup \{u\}$. 
Similarly, for a set of vertices $S\subseteq V$, we define $N(S) = \bigcup_{u\in S}N(u)\setminus S$  and $N^+(S) = \bigcup_{u\in S}N^+(u) = N(S)\cup S$.
For two (not necessarily disjoint) sets $S,T\subseteq V$, we let $E(S,T) = \{(u,v)\in E\colon u\in S, v\in T\}$ be the set of edges with one endpoint in $S$ and the other in $T$.
We also define $E(S) = E(S,S)$. 
By $G[S]$ we denote the induced subgraph of $G$ on $S\subseteq V$, i.e., $G[S] = (S, E(S))$.

\section{The 2-State and 3-State MIS Processes}
\label{sec:2-state-3-state-def}
We define two self-stabilizing distributed graph processes that compute a maximal independent set when applied on any given graph.

\begin{definition}[2-State MIS Process]
    \label{def:2-state-MIS}
    In the \emph{$2$-state MIS process} on graph $G=(V,E)$, each vertex $u\in V$ has a binary state from  the set $\{\black,\white\}$, and all states are updated in parallel rounds.  
    The initial state $c_0(u)$ of vertex $u$ can be arbitrary, and in each round $t=1,2,\ldots,$ $u$'s state is updated from $c_{t-1}(u)$ to $c_t(u)$ according to the following rule.
    
    \begin{algorithm}[H] {\small
        \DontPrintSemicolon
        let $\mathit{NC}_t(u) = \{c_{t-1}(v)\colon v\in N(u)\}$\;
        \uIf{$\big(c_{t-1}(u) = \black$ {\bf and} 
        $\mathit{NC}_t(u) \ni \black
        \big)$ 
        {\bf or} 
        $\big(c_{t-1}(u) = \white$ {\bf and} 
        $\mathit{NC}_t(u)\not\:\!\ni \black
        \big)$}{
            let $c_t(u)$ be a uniformly random state from $\{\black,\white\}$\;}
        \lElse{set $c_t(u) = c_{t-1}(u)$}
    }
    \end{algorithm}
\end{definition}


We say that vertex $u$ is \emph{\textbf{black}} or \emph{\textbf{white}} 
if its state is $\black$ or $\white$, respectively.
We say that $u$ is \emph{\textbf{active}} if 
it is black and has some black neighbor, or it is white and has no black neighbors.

We say that vertex $u$  
is \emph{\textbf{stable}}, if either it is black and has no black neighbors, or it is white and has a neighbor that is black \emph{and stable}.
It is immediate from the update rule that once a vertex becomes stable, it remains stable thereafter, and its state no longer changes.
The \emph{\textbf{stabilization time of vertex $u$}} is the earliest round at the end of which $u$ is stable.
The \emph{\textbf{stabilization time of the process}} is the earliest round at the end of which all vertices are stable.
It is easy to verify that after the stabilization time of the process, the set of black vertices is an MIS of $G$.  

We let $B_t = \{u\in V\colon c_t(u) = \black\}$ be the set of  black vertices at the end of round $t\geq 0$, and let $W_t = V\setminus B_t$ be the set of white vertices.
We let $$A_t = \{u\in B_t \colon N(u) \cap B_t \neq \emptyset \} \cup \{u\in W_t \colon N(u) \cap B_t = \emptyset \}$$ denote the set of active vertices at the end of round $t$.
We let $I_t = \{u \in B_t\colon N(u)\cap B_t = \emptyset\}$ be the set of stable black vertices at the and of round $t$  (note that $I_t$ is an independent set and is a subset of the final MIS).  
Finally, we let $V_t = V \setminus N^+(I_t)$ be the set of vertices that are not stable at the end of round $t$.

\begin{definition}[3-State MIS Process]
    \label{def:3-state-MIS}
    In the \emph{$3$-state MIS process} on $G=(V,E)$, each vertex $u\in V$ has a state from  set $\{\blkone, \blkzero,\white\}$, and the states are updated in parallel rounds.  
    The initial state $c_0(u)$ of $u$ is arbitrary, and in each round $t\geq 1,$ $u$'s state is updated as follows.
    
    \begin{algorithm}[H] {\small
        \DontPrintSemicolon
        let $\mathit{NC}_t(u) = \{c_{t-1}(v)\colon v\in N(u)\}$\;
        \uIf{$c_{t-1}(u) = \blkone$
        {\bf or}
        $\big(c_{t-1}(u) = \blkzero$ {\bf and} 
        $\mathit{NC}_t(u)\not\:\!\ni \blkone
        \big)$
        {\bf or} 
        $\big(c_{t-1}(u) = \white$ 
        \\$\,$
        {\bf and} 
        $\mathit{NC}_t(u)=\{\white\}
        \big)$}{
            let $c_t(u)$ be a uniformly random state from $\{\blkone,\blkzero\}$\;}
        \uElseIf{$c_{t-1}(u) = \blkzero$}
            {set $c_t(u) = \white$\;}
        \lElse   
            {set $c_t(u) = c_{t-1}(u)$}
    }
    \end{algorithm}
\end{definition}

In the 3-state MIS process, we say that a vertex $u$ is \emph{black} when its state is $\blkone$ or $\blkzero$.
Then the \emph{stable} vertices and the \emph{stabilization times} are defined as before. 
Note that the state of a stable black vertex alternates perpetually between states $\blkone$ and $\blkzero$.

In this paper we focus on the 2-state MIS process, but we expect that all our upper bound results should carry over to the 3-state MIS process.

\subsection{Basic Properties of the 2-State MIS Process}
\label{sec:basic-properties-2s}

We show some elementary properties of the 2-state MIS process.
In the analysis, it will be convenient to assume that at the beginning of each round $t\geq 1$, we flip for each vertex $u$ an independent coin $\phi_t(u)$ such that $\Pr[\phi_t(u) = \black]=\Pr[\phi_t(u)=\white]=1/2$. 
Then if $u$ must update its state to a random state in that round, i.e., if $u\in A_{t-1}$, we set $c_t(u) = \phi_t(u)$; while if $u\notin A_{t-1}$, then $\phi_t(u)$ is not used by the algorithm.

The lemmas below apply for any graph 
$G=(V,E)$, 
and the probabilistic statements assume that we know the states of vertices at the end of round $t$ (i.e., $B_t$ or $W_t$ is given). 
The first lemma says than an active vertex $u$ with $k$ active neighbors has probability $\Omega(1/k)$ to become stable black in the next $O(\log k)$ rounds.

\begin{lemma}
\label{lem:active-vertex-st}
  If $u\in A_t$ and $|N(u) \cap A_t|  = k\geq 1$, 
  then the probability that $u\in I_{t+\log(k+1)}$ is at least $(2ek)^{-1}$.  
\end{lemma}
\begin{proof}
  Let $r = \lceil\log(k+1)\rceil$.
  The probability that $u\in I_{t+r}$ is lower bounded by the probability that  $\phi_{t+1}(v)=\dots=\phi_{t+r}(v)=\black$ holds for $v=u$ and does not hold for any  $v\in N(u)\cap A_t$, which is
  \begin{equation}
    \label{eq:coins_2}
    (1/2)^r\cdot \left(1-(1/2)^r\right)^k
    \geq
    (1/2)^r\cdot e^{-k/(2^r-1)}
    \geq
    ({1}/{2k})\cdot ({1}/{e})
    .
  \end{equation}
  For the first inequality we used the fact $(1-1/n)^{n-1}\geq e^{-1}$, and for the second we used that $\log(k+1)\leq r\leq \log(k)+1$. 
\end{proof}

The next statement is a generalization of \cref{lem:active-vertex-st} to multiple active vertices $u_1,\ldots,u_\ell$.
We will apply this result to the set of active neighbors of a vertex $u$, to lower bound the probability that $u$ is stable after a logarithmic number of rounds (because a neighbors becomes stable black).
The proof can be found in \cref{sec:proof-neighbors2-st}.

\begin{lemma}\label{lem:neighbors2-st}
  Suppose that $u_1,\ldots,u_\ell\in A_t$, and
  $|N(u_i) \cap A_t|  = k_i > 0$, for each $1\leq i\leq \ell$.
  Then the probability that $\{u_1,\ldots,u_\ell\}\cap I_{t+\log(\max_i k_i+1)}\neq \emptyset$ is at least $(1/5)\cdot\min\left\{1,\, \sum_i(2k_i)^{-1}\right\}$.
\end{lemma}

\section{Simple Bounds for the 2-State MIS Process}
\label{sec:simple-bounds}

We show some simple bounds on the stabilization time of the 2-state MIS process on certain graph families, namely, the complete graph and trees (or more generally, graphs of bounded arboricity). 
We also show a basic upper bound in terms of the maximum degree on a general graph.

\begin{theorem}
    \label{thm:clique}
    The stabilization time of the 2-state MIS process on the complete graph $K_n=(V,E)$ is $O(\log n)$ in expectation and $O(\log^2n)$ w.h.p. 
    More concretely, for any $k>0$, the stabilization time is at least $k\cdot \log n$ with probability $2^{-\Theta(k)}$. 
\end{theorem}


\begin{proof}
    We call round $t$ \emph{critical} if $|B_t|\leq 1$, and we call it \emph{stable} if $|B_t| = 1$.
    Let $p_a$ be the probability that the next critical round is stable, given that $|A_t| = a \geq 2$.
    Note that in graph $K_n$, $A_t= B_t$ if $|B_t|>1$, $A_t = \emptyset$ if $|B_t|=1$, and $A_t=V$ if $B_t = \emptyset$; thus $|A_t| \neq 1$.
    We argue that for any $a\geq 2$,
    \[
        2/3\leq p_a\leq 17/21 
        .
    \]
    The lower bound follows from the observation that, for any $i\geq 2$ and $j\geq1$, the conditional probability that round $j$ is stable, given that it is critical and that $|A_{j-1}| = i$, is $\frac{\binom{i}{1}2^{-i}}{\binom{i}{1}2^{-i} + 2^{-i}} = \frac{i}{i+1} \geq 2/3$, since $i\geq 2$.
    For the upper bound we observe that, for any $i\geq 3$ and $j\geq1$, the conditional probability of $|B_{j}| \in \{2,0\}$, given that $|B_{j}| \leq 2$ and that $|A_{j-1}| = i$, is 
    $\frac{\binom{i}{2}2^{-i} + 2^{-i}}{\binom{i}{2}2^{-i} + \binom{i}{1}2^{-i} + 2^{-i}} 
    = 
    \frac{i^2-i + 2}{i^2 + i +2} \geq 4/7$.    
    Also, $p_{2} = 2/3<17/21$.
    Then, for any $a\geq 3$, we have
    $1-p_a \geq (4/7)\cdot(1-p_2)$, which implies
    $p_a\leq 17/21$.

    Next, consider the number of rounds $r$ from a non-stable critical round (when all nodes are white) until the next critical round.
    The probability that $r>k$ is lower and upper bounded by
    \[
        1 - e^{-n2^{-k}} \leq 1-(1-2^{-k})^n \leq n 2^{-k}
        .
    \]

    Combining the above we obtain that (i) from any given non-stable round, the probability that a stable round is reached in at most $k = \log n + 1$ rounds is at least $2/3 - n 2^{-k} \geq 1/6$; (ii)~from any given non-stable critical round, the probability that the next critical round is non-stable and is reached in more than $k = \log n - 2$ rounds is at least $1 - 17/24 - e^{-n2^{-k}} > 1/6$; and (iii)~assuming round $t=0$ is not critical, the probability that the first critical round is non-stable is at least $1 - 17/24$.
    These statements, together, imply that the stabilization time is at least $k\log n$ with probability $2^{-\Theta(k)}$.
    And from that, the expectation and high-probability bounds follow.
 \end{proof}

\begin{remark}
    \label{rem:union-of-cliques}  
     From \cref{thm:clique}, it is immediate that the \emph{expected} stabilization time of the 2-state MIS process is $\Theta(\log^2 n)$ on a graph $G$ that is the disjoint union of $\sqrt{n}$ cliques $K_{\sqrt{n}}$.
     The same bound holds also w.h.p.
\end{remark}

\begin{remark}
    \label{rem:clique-3state}    
    A similar analysis as for \cref{thm:clique} gives an upper bound of $O(\log n)$ on the stabilization time of the \emph{3-state} MIS process on $K_n$, both in expectation and \emph{w.h.p.}
    The reason is that once $B_t\neq\emptyset$ then $B_{t'}\neq\emptyset$ for all $t'\geq t$ (thus the next critical round is stable).
\end{remark}

\begin{theorem}
    \label{thm:trees}
    The stabilization time of the 2-state MIS process on any graph $G=(V,E)$ of bounded arboricity (e.g., $G$ is a tree)
    is $O(\log n)$ w.h.p. 
\end{theorem}
\begin{proof}
    Recall that the arboricity $\lambda$ of $G$ is the minimum number of forests into which its edges can be partitioned, and is equal up to a factor of 2 to the maximum average degree in any subgraph~\cite{Nash-Williams64}.
    Suppose that the average degree of any subgraph of $G$ is at most $d \leq 2\lambda$.
    Let $S_t$ be the subset of $V_t$ consisting of of all vertices $u\in V_t$ with $|N(u)\cap V_t| \leq d$. 
    Then $|S_t| \geq |V_t|/(d+1)$.
    If $u\in S_t\setminus A_t$ and $|N(u)\cap V_t| = d_u$, the probability that $N(u)\subseteq W_{t+1}$ is $2^{-d_u}\geq 2^{-d}$.
    Thus, for each $u\in S_t$, the probability that $u\in A_t\cup A_{t+1}$ is at least $2^{-d}$.
    And if $u\in A_t\cup A_{t+1}$, \cref{lem:active-vertex-st} gives that $u\in I_{t+\log(d+1)+1}$ with probability at least $(2ed)^{-1}$.
    It follows 
    \[
        \Exp\left[|V_{t+\log(d+1)+1}| \ \middle|\ |V_t|\right] 
        \leq
        |V_t| - (2ed)^{-1}\cdot 2^{-d}\cdot |V_t|/(d-1)
        \leq
        (1-\epsilon)\cdot |V_t|,
    \]
    for some constant $\epsilon = \epsilon(d)$.
    Let $r = \log(d+1)+1$. 
    Applying the above inequality iteratively, we obtain $\Exp[|V_{rt}|] \leq (1-\epsilon)^r n \leq e^{-\epsilon r}n$. 
    Thus for $t = 3\epsilon^{-1}\ln n$, $\Exp[|V_{rt}|] \leq n^{-2}$, and by Markov's inequality, $\Pr[|V_{rt}| \geq 1] \leq n^{-2}$, which implies the lemma.
    \end{proof}


\begin{theorem}
    \label{thm:UB-delta}
    The stabilization time of the 2-state MIS process on any graph $G=(V,E)$ of maximum degree $\Delta$ is at most $O(\Delta\log n)$ w.h.p. 
\end{theorem}

\begin{proof}
    We observe that if $u\in V_t$ then $N^+(u)\cap A_t \neq \emptyset$.
    Let $u\in V_0$, and let $(v_1,t_1),(v_2,t_2),(v_3,t_3),\ldots$ be a random sequence of vertex-round pairs defined as follows:
    Let $t_0 = 0$.
    For each $i\geq 1$, if $u \in V_{t_{i-1}}$, then $v_i$ is an arbitrary vertex from the set $N(u)\cap A_{t_{i-1}}$, and
    $
        t_i = 
        \min 
        \{j>t_{i-1}\colon v_i\notin A_j\}
        ;
    $
    while if $u \notin V_{t_{i-1}}$, then $(v_i,t_i) = (u,t_{i-1})$.
    
    We focus on the first $r = 6e\Delta\log n$ elements of the sequence above.
    We bound the probability that $u\in V_{t_r}$.
    For each $1\leq i \leq r$, the conditional probability that $v_i \in I_{t_{i+1}}$ (and thus $u\notin V_{t_{i+1}}$), given $v_i$ and $B_{t_i}$, is at least $1/(2e\Delta)$, from \cref{lem:active-vertex-st}.    
    It follows that 
    \[
        \Pr[u\in V_{t_r}] 
        \leq
        (1-1/(2e\Delta))^r
        \leq
        e^{-r/(2e\Delta)}
        =
        n^{-3}
        .
    \]
    Next, we bound the value of $t_r$.
    For each $1\leq i \leq r$ and $t\geq t_{i-1}$, if $v_i\in A_t$ then the conditional probability that $v_i \notin A_{t+1}$, given $(v_i,t_i)$ and $B_t$, is exactly $1/2$ (in all cases).
    It follows that the probability of $t_r > 4r$ is upper bound by the probability that a sequence of $4r$ fair coin tosses contains fewer than $r$ heads. 
    Thus, by a Chernoff bound,
    \[
        \Pr[t_r > 4r]
        \leq
        e^{-(1/2)^2 2r/2}
        =
        e^{-e\Delta \log n}
        <
        n^{-3}
        .
    \]
    Combining the above results, we obtain that 
$
        \Pr[u\notin V_{4r}]
        \geq
        \Pr[\{u \notin V_{t_r}\}\cap\{t_r\leq 4r\}]
        \geq
        1- 2n^{-3}
        .
 $
    (Recall that $r=6e\Delta\log n$.)    
    Finally, a union bound over all $u\in V$ competes the proof.
\end{proof}

\section{The 2-State MIS Process 
on Random Graphs} 
\label{sec:2-state-random-graphs}

We first show some additional properties of the 2-state MIS process, which hold for any graph but are useful only when adjacent vertices do not have many common neighbors.
Then we show some structural properties of $G_{n,p}$ random graphs.
Finally, we use these properties to show a $\poly(\log n)$ upper bound on the stabilization time of the 2-state MIS process on $G_{n,p}$ random graphs. 

\subsection{Refined Properties of the 2-State MIS Process}
\label{sec:refined}


We call a vertex \emph{\textbf{$k$-active}} if it is active and has at most $k$ active neighbors. 
Let 
\[
    A^k_t = 
    \{u\in A_t\colon |N(u)\cap A_t| \leq k\}
\]
be the set of $k$-active vertices at the end of round $t$.
From \cref{lem:active-vertex-st}, a $k$-active vertex has probability at least $\Omega(1/k)$ to become stable black in the next $O(\log k)$ rounds. 
It is thus desirable to have $k$-active vertices for small values $k$.

In this section we establish lower bounds on the probability that a given vertex $u$ becomes $k$-active at some point in the next $r$ rounds, as a function of the probability that $u$ is active (but has possibly more than $k$ active neighbors) 
at a point in a certain subinterval of those $r$ rounds.

The next key lemma is the base of all the other results in the section. 
It lower bounds the probability $q$ of a white vertex $u$, which is non-active and non-stable, to become $k$-active  after a single round.
The lower bound is expressed in terms of the probability $p$ 
\gtodo{FUTURE-TODO: change $p$ to a different symbol; $p$ is used in $G_{n,p}$}
that $u$ is active white
after two rounds.
The value of $k$ depends on the number of active neighbors of $u$, and, crucially, on the number of their common neighbors with $u$.
\gtodo{FUTURE-TODO: minor changes needed to carry over to 3-color process: $u\in W_t$; at the beginning of proof, some white are grey; some coin-flips are grey; in general some `white' must be `non-black' instead}

\begin{lemma}
    \label{lem:key-lemma}
    Suppose that $u\in V_t\setminus A_t$,\footnote{Note that $u\in V_t\setminus A_t$ implies $u\in W_t\cap W_{t+1}$.} 
    and let $\theta = |N(u)\cap N^+(A_t\cap N(u))|$ be the number of $u$'s neighbors that are active or adjacent to an active neighbor of $u$ at the end of round $t$.
    Let $p$ be the probability that 
    $u\in A_{t+2}\cap W_{t+2}$, 
    and $q$ the probability that $u\in A_{t+1}^k$ where  $k =  \theta + \lceil\log (1/p)\rceil$. 
    Then $q\geq p^\alpha$, where $\alpha = {1}/{\log(4/3)}\leq 2.41$.
\end{lemma}
\begin{proof}
    Let $D = N(u)\cap A_t$.
    In round $t+1$, each $v\in D$ updates its state to a random state, while each $v\in N(u)\setminus D$ remains white. 
    Let $Z = N(u)\cap A_{t+1}\setminus N^+(D) $ be the set of active neighbors of $u$ at the end of round $t+1$ that are at distance at least two away from set $D$.
    Clearly, $Z$ does not depend on the random choices of vertices $v\in D$ in round $t+1$.

    We have that $u\in A_{t+1}$ \emph{if and only if} all $v\in D$ update their state to white in round $t+1$, i.e., $\phi_{t+1}(v) = \white$.\footnote{Recall the discussion about coin flips $\phi_{t}(v)$ at the beginning of \cref{sec:basic-properties-2s}.}
    Also $|N(u)\cap A_{t+1}| \leq |N(u)\cap N^+(D)| + |Z|= \theta + |Z|$.
    It follows
    \[
        q \geq \left(1/2\right)^d\cdot\Pr[|Z|\leq \lambda],
    \]
    where $d = |D|$ and  $\lambda = \lceil\log (1/p)\rceil$.

    We have that $u\in A_{t+2}\cap W_{t+2}$ \emph{only if} $\phi_{t+1}(v)$ or $\phi_{t+2}(v) = \white$ for every $v\in D$, and $\phi_{t+2}(v) = \white$ for every $v\in Z$.
    It follows that
    \begin{equation}
        \label{eq:3over4}
        p
        \leq \left(3/4\right)^d\cdot\sum_{i\geq0}\Pr[|Z| = i]/2^i
        .
    \end{equation}
    Let $\varepsilon = \Pr[|Z|\leq \lambda]$.
    Then
    \[
        p
        \leq 
        \left(3/4\right)^d\cdot
        \left(
        \varepsilon
        +
        (1-\varepsilon)/2^{\lambda+1}
        \right)
        \leq 
        \left(3/4\right)^d\cdot
        \left(
        \varepsilon
        +
        (1-\varepsilon)\cdot p/2
        \right).
    \]
    This implies that $p\leq \varepsilon +(1-\varepsilon)\cdotp p/2$, thus $p \leq 2\varepsilon/(1+\varepsilon)$, and substituting that above yields
    \[
        p
        \leq
        \left(3/4\right)^d\cdot
        \left(
        \varepsilon
        +
        (1-\varepsilon)\cdot\varepsilon/(1+\varepsilon)
        \right)
        =
        \left(3/4\right)^d\cdot
        \frac{2\varepsilon}{1+\varepsilon}
        .
    \]

    Finally, since $\left(3/4\right)^{d\alpha} = \left(1/2\right)^{d}$, 
    and for all $x\in [0,1]$,
    $\left(\frac{2x}{1+x}
        \right)^\alpha 
    \leq 
    \left(\frac{2x}{1+x}
        \right)^2
    =
    x\cdot \frac{4x}{(1+x)^2}
    \leq x$, 
    \[
        p^\alpha
        \leq
        \left(3/4\right)^{d\alpha}\cdot
        \left(
        \frac{2\varepsilon}{1+\varepsilon}
        \right)^\alpha
        \leq
        \left(1/2\right)^{d}
        \cdot
        \varepsilon
        \leq q.
        \qedhere
    \]
\end{proof}

Next, we use the above \cref{lem:key-lemma} to prove a similar result over a sequence of $r$ rounds.
For any vertex $u\in V$ and $i\geq 1$, let
\begin{equation}
    \label{eq:theta}
    \theta_u(i) = \max\{|N(u)\cap N^+(S)| \colon S\subseteq N(u),\, |S|\leq i\}
    . 
\end{equation}

\begin{lemma}
    \label{lem:non-active-to-pretty}
    Suppose that $u\in V_t\setminus A_t$ 
    \gtodo{FUTURE-TODO: minor changes needed for 3-color process: $u\in W_t$}
    and let $d = |N(u)\cap A_t|$.
    Let $p_r$ be the probability that $u \in A_{t+1}\cup\dots\cup A_{t+r}$, and let $q_r$ be the probability that $u\in A^k_{t+1}\cup\dots\cup A^k_{t+r-1}$, where 
    $$k = \theta_u\left(\alpha \log\left(\tfrac{4r}{p_r-2^{-d}}\right)\right)+\left\lceil \log\left(\tfrac{4r}{p_r-2^{-d}}\right)\right\rceil,
    $$
    and $\alpha = 1/{\log(4/3)}$.
    Then, for any $r\geq 2$, $q_{r}\geq r^{1-\alpha}\cdot \left(\frac{p_r-2^{-d}}{2}\right)^\alpha$.
\end{lemma}


 \begin{proof}
    For $i\geq 0$, let $d_i = |N(u) \cap A_{t+i}|$, and define the following events:
    $\WW_i$ is the event that $u\in W_{t+i}$;
    $\AA_i$ is the event that $u\in A_{t+i}$;
    $\AA^k_i$ is the event that $u\in A^k_{t+i}$; and $\HH_i=\bar \AA_0\cap \bar \AA_1 \cap\dots\cap \bar \AA_i$. 
    Let also $\XX_i$ be the event that 
    the states of the vertices at the end of round $t+i$ are such that the conditional probability of $\AA_{i+2}\cap W_{i+2}$ is at least $\frac{p_r-p_1}{4r}$. 
    Let $r\geq 2$ and $\lambda = \lfloor\alpha \log\left(\frac{4r}{p_r-p_1}\right)\rfloor$.
    Then
    \begin{align*}
        p_r
        &=
        \sum_{1\leq i\leq r} \Pr[\AA_i \cap \HH_{i-1}]
        \\&
        =
        p_1 +
        \sum_{2\leq i\leq r} \Pr[\AA_i \cap \HH_{i-1}]
        \\&
        =
        p_1 +
        \sum_{2\leq i\leq r} \Pr[\AA_i\cap\WW_i \cap \HH_{i-1}]
        \text{\quad(since $\AA_i\cap \HH_{i-1}$ implies $\WW_i$)}
        \\&
        \leq
        p_1 +
        \sum_{2\leq i\leq r} \Pr[\AA_i\cap\WW_i \cap \HH_{i-2}]
        \text{\quad(since $\HH_{i-1}$ implies $\HH_{i-2}$)}
        \\&
        \leq
        p_1 +
        \sum_{2\leq i\leq r} \Pr[\AA_i\cap\WW_i \cap \HH_{i-2}\cap \{d_{i-2} \leq \lambda \}\cap \XX_{i-2}]
        \\&\qquad\
        +
        \sum_{2\leq i\leq r} \Pr[\AA_i\cap\WW_i \cap \HH_{i-2}\cap \{d_{i-2} > \lambda \}]
        +
        \sum_{2\leq i\leq r} \Pr[\AA_i \cap\WW_i \cap \bar \XX_{i-2}]
        .
    \end{align*}
    Each of the last two sums above is at most $\frac{p_r-p_1}{4}$, because for each non-zero sum term, we have
    \[
        \Pr[\AA_i\cap\WW_i \cap \HH_{i-2}\cap \{d_{i-2} > \lambda \}]
        \leq
        \Pr[\AA_i\cap\WW_i \mid \HH_{i-2},\, d_{i-2} > \lambda]
        \leq
        \left(\frac34\right)^{\lambda+1}
        \leq \frac{p_r-p_1}{4r}
        ,
    \]
    similarly to \cref{eq:3over4},
    and 
    $
        \Pr[\AA_i \cap\WW_i\cap \bar \XX_{i-2} ]
        \leq
        \Pr[\AA_i\cap\WW_i \mid  \bar \XX_{i-2} ]
        \leq
        \frac{p_r-p_1}{4r}
        .
    $
    Applying these above gives
    \begin{align*}
        \frac{p_r-p_1}{2}
        &\leq
        \sum_{2\leq i\leq r} \Pr[\AA_i\cap\WW_i \cap \HH_{i-2}\cap \{d_{i-2} \leq \lambda \}\cap \XX_{i-2}]
        \\&
        =
        \sum_{2\leq i\leq r} \Pr[\AA_i\cap\WW_i \mid \HH_{i-2},\,  d_{i-2} \leq \lambda,\ \XX_{i-2}]
        \cdot\Pr[\HH_{i-2}\cap \{d_{i-2} \leq \lambda \}\cap \XX_{i-2}]
        .
    \end{align*}
    Next we lower bound $q_{r}$. 
    We have
    \[
      \begin{aligned}
        q_{r} 
        &\geq
        \sum_{1\leq i\leq r-1} \Pr[\AA^k_i \cap \HH_{i-1}]
        \\&
        =
        \sum_{2\leq i\leq r} \Pr[\AA^k_{i-1} \cap \HH_{i-2}]
        \\&
        \geq
        \sum_{2\leq i\leq r} \Pr[\AA^k_{i-1} \cap \HH_{i-2}\cap \{d_{i-2} \leq \lambda\}\cap \XX_{i-2}]
        \\&
        =
        \sum_{2\leq i\leq r} \Pr[\AA^k_{i-1} \mid \HH_{i-2},\,  d_{i-2} \leq \lambda,\, \XX_{i-2}]
        \cdot\Pr[\HH_{i-2}\cap \{d_{i-2} \leq \lambda \}\cap \XX_{i-2}]
        .
      \end{aligned}
    \]
    From \cref{lem:key-lemma},
    applied for round $t+i-2$, using $p \geq \frac{p_r-p_1}{4}$ and $\theta \leq \theta_u(\lambda)$,  and observing that $p_1 = 2^{-d}$, we obtain
    \[
        \Pr[\AA^k_{i-1} \mid \HH_{i-2},\,  d_{i-2} \leq \lambda, \XX_{i-2}] \geq \left(\Pr[\AA_i\cap\WW_i \mid \HH_{i-2},\,  d_{i-2} \leq \lambda, \XX_{i-2}]\right)^\alpha
        .
    \]
    We substitute this to the previous equation above, and then use Jensen's inequality to complete the proof:
    Let $\nu = \sum_{2\leq i\leq r} \Pr[\HH_{i-2}\cap \{d_{i-2} \leq \lambda \}\cap \XX_{i-2}] \leq r$.
    \begin{align*}
        q_{r} 
        &\geq
        \sum_{2\leq i\leq r} \left(\Pr[\AA^k_i \mid \HH_{i-2},\,  d_{i-2} \leq \lambda,\,\ \XX_{i-2}]\right)^\alpha
        \cdot\Pr[\HH_{i-2}\cap \{d_{i-2} \leq \lambda \}\cap \XX_{i-2}]
        \\&
        \geq \nu\cdot
        \left(\sum_{2\leq i\leq r} \Pr[\AA^k_i \mid \HH_{i-2},\,  d_{i-2} \leq \lambda,\, \XX_{i-2}]
        \cdot\Pr[\HH_{i-2}\cap \{d_{i-2} \leq \lambda \}\cap \XX_{i-2}]/\nu\right)^\alpha
        \\&
        \geq \nu\cdot
        \left(\frac{p_r-p_1}{2\nu}\right)^\alpha
        \geq r\cdot
        \left(\frac{p_r-2^{-d}}{2r}\right)^\alpha
        .
        \qedhere
    \end{align*}
\end{proof}

\cref{lem:non-active-to-pretty} assumes that vertex $u$ is initially not active.
\gtodo{FUTURE-TODO: next two lemmas are probably not necessary; but the previous two are good -- are they?}
\itodo{Yes the first is good. Maybe the second (previous) one can be deleted in some way, but let's keep it and think about it for the journal version.}
The next lemma shows a similar result for the case where $u$ is active initially.
In this case, in place of the probability $p_r$ that $u$ becomes active at some point in the interval $\{t+1,\ldots, t+r\}$, we use the probability $b_r$ that $u$ becomes black  at some point of a subinterval $\{t+\ell,\ldots, t+r\}$.
The proof proceeds by considering the first round after $t$ when either $u$ has at most $k$ black neighbors, or $u$ is white.
If the first condition holds, then $u$ has probability $1/2$ of being black, and thus of being $k$-active. 
If only the second condition holds then we are in the case of \cref{lem:non-active-to-pretty}.
The proof can be found in \cref{sec:proof-of-active-to-pretty}.

\begin{lemma}
    \label{lem:active-to-pretty}
    Suppose that $u\in A_t$.
    Let $\ell \geq 2$ and $r \geq \ell +2$, let $b_{r}$ be the probability 
    that $u \in B_{t + \ell}\cup\dots\cup B_{t+r}$, and suppose that $b_r \geq 1/2^{\ell-2}$. 
    Let $q_r$ be the probability 
    that $u\in A^k_t\cup\dots\cup A^k_{t+r-1}$, where 
    \[
        k = \theta_u\big(\alpha \log\left({32r}/{b_r}\right)\big)
        +
        \log\left({32r}/{b_r}\right)
        + \log(1/b_r)+3
        .
    \]
    Then $q_{r} \geq r^{1-\alpha} \cdot \left({b_r}/{16}\right)^\alpha$,
    where $\alpha = {1}/{\log(4/3)}$.    
\end{lemma}
    
In the last lemma of this section, we consider the case in which \cref{lem:non-active-to-pretty} does not give a large enough lower bound for $q_r$, even though $p_r$ is large, because the difference $p_r - 2^{-d}$ is small.
We proceed by essentially reducing this case to the case of \cref{lem:active-to-pretty}, after a single round.
The proof is in \cref{sec:proof-non-active-to-pretty2}.

\begin{lemma}
    \label{lem:non-active-to-pretty2}
    Suppose that $u\in V_t \setminus A_t$, and let $d = |N(u)\cap A_t|$.
    Let $\ell\geq 5$ and $r\geq \ell + 2$, let $p_{r}$ be the probability that $u \in A_{t+1}\cup\dots\cup A_{t+r-1}$, let $b_r$ be the probability that $u \in B_{t+\ell}\cup\dots\cup B_{t+r}$, and suppose that $b_r \geq 1/2^{\ell-4}$ and $b_r \geq 2(p_r - 2^{-d})$.
    Let $q_r$ be the probability 
    that $u\in A^k_t\cup\dots\cup A^k_{t+r-1}$, where 
    \[
        k = \theta_u\big(\alpha \log\left({128r}/{b_r}\right)\big)
        +
        \log\left({128r}/{b_r}\right)
        + \log(4/b_r)+3
        .
    \]
    Then $q_{r} \geq r^{1-\alpha} \cdot \left({b_r}/{64}\right)^\alpha$,
    where $\alpha = {1}/{\log(4/3)}$.    
\end{lemma}

\subsection{Structural Properties of \texorpdfstring{$G_{n,p}$}{Gnp }
and Good Graphs}
\label{sec:structural-properties}

We describe some structural properties that a graph must possess in order for the analysis given in the following sections to carry through.
A graph satisfying these properties is called a \emph{good graph}.
Then we show that a random $G_{n,p}$ graph is a good graph w.h.p.

\begin{definition}[Good Graphs]
    \label{def:good-graph}
    Let $n$ be a positive integer and $0<p<1$.
    A graph $G=(V,E)$ with $n$ vertices is \emph{$(n,p)$-good} if it satisfies all the following properties:
    \begin{enumerate} [(P1)]
    
        
        \item\label{def:good-degree}
        For any set $S\subseteq V$, the average degree of induced subgraph $G[S]$ is at most $\max\{8p|S|,\, 4\ln n\}$.

        \item\label{def:good-subset-neighbors}
        For any set $S\subseteq V$ of size $|S| \geq 40\ln (n)/p$, 
        $$|\{u\in V\setminus S \colon |N(u)\cap S| < p|S|/2\}|\leq |S|/2.$$

        \item\label{def:good-STI}
        For any three disjoint sets $S,T,I \subseteq V$ such that $|S| \geq 2|T|$ and $(S\cup T) \cap N(I) = \emptyset$,
        \[
            |N(T) \setminus N^+(S\cup I)| 
            \leq|
            N(S)\setminus N^+(I)| 
            + 
            8\ln^2(n)/p 
            .  
        \]

        \item \label{def:good-ST}
        For any two disjoint sets $S,T\subseteq V$ such that $|S|\geq |T|$ and $|T|\leq \ln (n)/p$,
        $
            |E(S,T)| \leq 6|S|\ln n.
        $

        \item \label{def:good-common}
        No two vertices $u,v\in V$ have more than $\max\{6np^2,\, 4\ln n\}$ common neighbors.

        \item \label{def:good-common-lb}
        If $p \geq 2(\ln (n)/n)^{1/2}$ then $\diam(G) \leq 2$.        
    \end{enumerate}
\end{definition}

\begin{lemma}
    \label{lem:Gnp-good}
    A random graph $G=(V,E)$ drawn from $G_{n,p}$ is $(n,p)$-good with probability $1-O(n^{-2})$.  
\end{lemma}

The proof of \cref{lem:Gnp-good} can be found in \cref{sec:proof-Gnp-good}. 


\subsection{Analysis of the 2-State MIS Process on \texorpdfstring{$G_{n,p}$}{Gnp}}
\label{sec:2state-Gnp}

In this section, we prove the following bound on the stabilization time of the 2-state MIS process on a random $G_{n,p}$ graph. 

\begin{theorem}
    \label{thm:GnpUB-2s}
    The stabilization time of the 2-state MIS process on a random graph drawn from $G_{n,p}$, where $p = O(\sqrt{\log(n)/ n})$ or $p=\Omega(1/\log^{2.5} n)$, \itodo{changed} 
    is $O(\log^{5.5} n)$
    with probability $1-O(n^{-2})$.
\end{theorem}

The theorem follows by combining \cref{lem:Gnp-good} and the next lemma, which analyzes the 2-state MIS process on a good graph.

\begin{lemma}
    \label{lem:Good-2s}
    The stabilization time of the 2-state MIS process on any $(n,p)$-good graph $G=(V,E)$,  where $p = O(\sqrt{\log(n)/ n})$ or $p=\Omega(1/\log^{2.5} n)$, 
    \itodo{changed}
    is $O(\log^{5.5} n)$
    with probability $1-O(n^{-2})$.
\end{lemma}

It is straightforward to extend the above statements so that $p \leq \poly(\log n)\cdot n^{-1/2}$ or $p \geq 1/\poly(\log n)$, for any desired $\poly(\log n)$ term, by adjusting the exponent of $\log n$ in the stabilization time bound.\gtodo{new}

\subsubsection{
Proof of \texorpdfstring{\cref{lem:Good-2s}}{Lemma GG}}
\label{sec:2state-GG}

We show that starting from any vector of vertex states, the process makes sufficient progress after $\poly(\log n)$ rounds, where progress is measured by the expected number of vertices that become stable.
All lemmas below assume $G=(V,E)$ is an arbitrary $(n,p)$-good graph, and the probabilistic statements assume we know the states of the vertices at the end of round $t$. 
%
The first lemma considers the case in which the number of active vertices is large, namely, $|A_t|= \Omega(\log (n)/p)$.

\begin{lemma}\label{lem:GG-case-large-At}
  If $|A_t| \geq 80\ln(n)/p$ then there is a constant $\epsilon>0$ such that
  $
     \Exp[|V_{t+\log n}| ] \leq \left(1-
     \epsilon
     \right)\cdot |V_{t}|
     .
  $
\end{lemma}
\begin{proof} 
    From property~\cref{def:good-degree} in \cref{def:good-graph} of good graphs, the average degree of the induced subgraph $G[A_t]$ is at most $k = \max\{8p|A_t|,\, 4\ln n\} = 8p|A_t|$.
    
    Let $S$ be a subset of $A_t$ consisting of the  $|A_t|/2$ vertices $u\in A_t$ with the smallest degree in $G[A_t]$,
    i.e., for any two vertices $u\in S$ and $u'\in A_t\setminus S$, $|N(u)\cap A_t| \leq|N(u')\cap A_t|$.
    It follows that for all $u\in S$, 
    $|N(u)\cap A_t|\leq 2k$; thus $S \subseteq A_t^{2k}$.
    Let $R = \{u\in V\setminus S \colon |N(u)\cap S| < p|S|/2\}$. 
    Since $|S| = |A_t|/2 \geq 40\ln(n)/p$,  property~\cref{def:good-subset-neighbors} in \cref{def:good-graph} yields $|R| \leq |S|/2$.
    Then the number of vertices $u\in V_t$ with $|N(u)\cap S| \geq p|S|/2$ is at least
    \[
        |V_t\setminus (S\cup R)| 
        \geq |V_t| - (|S| + |R|) 
        \geq |V_t| - 3|S|/2 
        = |V_t| - 3|A_t|/4 
        \geq |V_t|/4
        .
    \]
    Since each of those vertices $u$ has at least $p|S|/2$ neighbors in $S\subseteq A_t^{2k}$, \cref{lem:neighbors2-st} gives that the probability at least one neighbor of $u$ is stable black (and thus $u$ is also stable) at the end of round $t+\log n$ is at least
    \[
        (1/5)\cdot\min\left\{1,\, (p |S|/2) \cdot (4k)^{-1}\right\}
        =
        (1/5)\cdot\min\left\{1,\, (p |A_t|/4) \cdot (32p|A_t|)^{-1}\right\}
        =
        1/640
        .
    \]
    Then the expected number of vertices that are not stable at the end of round $t + \log n$ is
    \[
        \Exp[|V_{t+\log n}|] 
        \leq 
        |V_t| - (|V_t|/4) \cdot 1/640
        \leq
        |V_t| -|V_t|/2560
        .
        \qedhere
    \]
\end{proof}


The next lemma considers the case in which the number of vertices that are not stable is large, namely $|V_t| = \Omega(\ln^2 (n)/p)$, and  $|A_t| = O(\ln(n)/p)$. 

\begin{lemma}\label{lem:GG-case-large-Vt}
  If $|V_t| \geq  10\ln^2(n)/p$
  and
  $|A_t|\leq 80\ln(n)/p$ 
  then 
  there is a constant $\epsilon>0$ such that
  $
     \Exp[|V_{t+\log n}| ] 
     \leq \left(1-
     \epsilon/\ln n
     \right)
     \cdot |V_t|
     .
  $
\end{lemma}

\begin{proof}
    From property~\cref{def:good-degree} in \cref{def:good-graph}, the average degree of graph $G[A_t]$ is at most 
    \[
        k = \max\{8p|A_t|,\, 4\ln n\} \leq 640\ln n.
    \]
    Let $S$ be a subset of $A_t$ consisting of the  $2|A_t|/3$ vertices $u\in A_t$ with the smallest degree in $G[A_t]$, and let $T = A_t\setminus S$.
    Then for all $u\in S$,
    $|N(u)\cap A_t|\leq 3k$; thus $S \subseteq A_t^{3k}$.
%

    The set $V_t$ 
    consist of (i)~all the active vertices, $u\in A_t = S\cup T$, and (ii)~all the non-active vertices that are not in $N^+(I_t)$
    (these vertices are white and have at least one active neighbor).
    We can thus partition $V_t$ into the four distinct sets: $S$, $N(S) \setminus N(I_t)$, $T\setminus N(S)$, and $N(T) \setminus N^+(S\cup I_t)$.
    For the sizes of these sets, we have
    $|T\setminus N(S)| \leq |T| < |S|$ and,
    by property~\cref{def:good-STI} in \cref{def:good-graph},
    \[
        |N(T) \setminus N^+(S\cup I_t)|
        \leq
        |N(S) \setminus N(I_t)| + 8\ln^2(n)/p
        .
    \]
    Using these two inequalities, the fact that the sizes of the four sets above sum to $|V_t|$, and the assumption $|V_t| \geq  10\ln^2(n)/p$, we obtain
    \[
        |S| + |N(S) \setminus N(I_t)| 
        \geq (|V_t| - 8\ln^2(n)/p)/2
        \geq |V_t|/10
        .
    \]
    Therefore, at least $|V_t|/10$ vertices $u\in V_t$ are in $S$ or adjacent to a vertex from $S$. 
    From \cref{lem:active-vertex-st},
    each $u\in S \subseteq A_t^{3k}$ is stable black (and all its neighbors are stable white) at the end of round $t+\log n$, with probability at least $1/(6ek)$.
    It follows that 
    \[
        \Exp[|V_{t+\log n}|] 
        \leq 
        |V_t| - (|V_t|/10) \cdot 1/(6ek)
        \leq
        |V_t| - |V_t|/(1.1\cdot 10^5\ln n )
        .
        \qedhere
    \]
\end{proof}

In the next lemma we analyze the remaining case, in which $|V_t| = O(\ln^2 (n)/p)$ and $|A_t| = O(\ln(n)/p)$. 
In fact, the lemma does not require a bound on $|A_t|$. 
Unlike the previous lemmas, however, it requires that $p = O(\sqrt{\log(n)/n})$.
\gtodo{FUTURE-TODO: make proof of \cref{lem:GG-case-small-AtVt} more similar to proof of \cref{lem:GG-case-small-AtVt-3c-small-p}}

\begin{lemma}\label{lem:GG-case-small-AtVt}
  If $|V_t| \leq  10\ln^2(n)/p$
  and $p \leq c \sqrt{\log(n)/ n}$, for some constant $c>0$,
  then 
  there is a constant $\epsilon = \epsilon(c) > 0$ such that
  $
     \Exp[|V_{t+2\log n}| ] 
     \leq \left(1-
     \epsilon/\ln^{3.5} n
     \right)
     \cdot |V_t|
     .
  $\footnote{If $c$ is super constant, then the proof gives $\Exp[|V_{t+2\log n}| ] 
     \leq \left(1-
     \epsilon /(c^2\ln^{3.5} n)
     \right)
     \cdot |V_t|$.
     }
\end{lemma}

\begin{proof}
    From property~\cref{def:good-degree} in \cref{def:good-graph}, the average degree of graph $G[V_t]$ is at most 
    \[
        k = \max\{8p|V_t|,\, 4\ln n\} \leq 80\ln^2 n.
    \]
    Let $T$ be a subset of $V_t$ consisting of the  $\min\{\ln(n)/p,\, |V_t|/2\}$ vertices $u\in V_t$ with the largest degree in $G[V_t]$, and let $S = V_t\setminus T$.
    Then $|S| \geq |T|$, and for all $u\in S$,  $|N(u)\cap V_t|$ is at most
    \[
        d =  k |V_t|/|T| \leq k \cdot \max\{p|V_t|/\ln n,\, 2\}
        \leq
        800\ln^3 n
        .
    \]
    From property \cref{def:good-ST} in \cref{def:good-graph}, the number of edges between $S$ and $T$ is
    $|E(S,T)| \leq 6|S|\ln n$.
    Let $R =  \{u\in S \colon |N(u)\cap T| \leq 12\ln n \}$.
    Then $|R| \geq |S|/2 \geq |V_t|/4$.
    We will show for some constant $\epsilon' = \epsilon'(c)$ that
    \begin{equation}
        \label{eq:uRstable}
        \Pr[u\notin V_{t+2\log n} ] 
        \geq 
        \epsilon' \ln^{-\alpha-1} n \cdot (\ln\ln n)^{-\alpha},\
        \text{ for all }
        u\in R.
    \end{equation}
    It follows that 
    $
        \Exp[|V_{t+2\log n}| ] 
        \leq 
        |V_t| - (|V_t|/4) \cdot \epsilon' \ln^{-\alpha-1} n \cdot (\ln\ln n)^{-\alpha}
        .
    $
    Since $\alpha = {1}/{\log(4/3)}\leq 2.41$, the above implies the lemma.
    To complete the proof it remains to show \cref{eq:uRstable}.    

    Let $u\in R$.
    We partition the neighbors of $u$ in $G[V_t]$ into sets $N(u)\cap S$ and $N(u)\cap T$, and let 
    \[
        x = \Pr[N(u)\cap S\cap A \neq\emptyset]
        \quad\text{and}\quad
        y = \Pr[N(u)\cap T\cap B \neq\emptyset]
        ,
    \]    
    where $A = A_t\cup\dots\cup A_{t+r-2}$, $B = B_{t+r-2}\cup B_{t+r-1}\cup B_{t+r}$, and $r = \log(48\ln n) + 6$. 
    We distinguish the following three cases: $x + y \leq 1/2$, $x\geq 1/4$, and $y \geq 1/4$.  

    \emph{Case $x+y \leq 1/2$}: 
    With probability at least $1-(x+y) \geq 1/2$, we have $N(u)\cap S\cap A  =\emptyset$ and $N(u)\cap T\cap B =\emptyset$.
    If $N(u)\cap S\cap A =\emptyset$ then $N(u) \cap S\subseteq W_{t+r-2}$  
    (it is easy to see that $N(u) \cap S \cap I_{t+r-2} = \emptyset$).
    Similarly, if $N(u)\cap T\cap B =\emptyset$, it is immediate that $N(u)\cap T\subseteq W_{t+r-2}$.
    Thus, with probability at least $1/2$, we have that $N(u) \subseteq W_{t+r-2}$. 
    If $N(u) \subseteq W_{t+r-2}$, then either $u\in A_{t+r-2}\cap W_{t-r-2}$ or $u\in I_{t+r-2}$.
    Therefore, with probability at least $1/2$, either $u\notin V_{t+r-2}$ or $u\in A_{t+r-2}$.
    If $u\in A_{t+r-2}$, then $u\in A^d_{t+r-2}$ since $|N(u)\cap V_t|\leq d$, and from \cref{lem:active-vertex-st}, the probability that  $u\in I_{t+r-2 +\log n}$ is at least $(2ed)^{-1}$.
    Combining the last two statements yields that the probability of $u\notin V_{t+r-2 +\log n}$ is at least $(1/2)\cdot(2ed)^{-1} \geq (8700\ln^3 n)^{-1}$,
    which implies  \cref{eq:uRstable}.
    
    \emph{Case $x \geq 1/4$}: 
    With probability at least $1/4$, there is a pair $v,j$ such that $v\in N(u)\cap S$, $0\leq j\leq r-2$, and  $v\in A_{t+j}$.
    And if $v\in A_{t+j}$ then $v\in A^d_{t+j}$ since $|N(v)\cap V_t|\leq d$, and from \cref{lem:active-vertex-st}, the probability that  $v \in I_{t+j +\log n}$ is at least $(2ed)^{-1}$.
    We conclude that the probability that $u \in N^+(I_{t+r-2+\log n})$ 
    is at least     $(1/4)\cdot(2ed)^{-1} \geq (17400\ln^3 n)^{-1}$, which implies \cref{eq:uRstable}.
    
    \emph{Case $y \geq 1/4$}: 
    There exists some $v^\ast \in N(u)\cap T$ such that
    \[
        \Pr[v^\ast \in B] 
        \geq 
        y/|N(u)\cap T|
        \geq 
        (4\cdot 12\ln n)^{-1}
        =
        (48\ln n)^{-1}
        .
    \]
    If $v^\ast \in A_t$ then we can apply
    \cref{lem:active-to-pretty},
    for $\ell = r- 2  \geq \log(48\ln n) + 2$ and $b_r \geq (48\ln n)^{-1}$, to obtain that $v^\ast \in A^\lambda_t\cup \dots\cup A^\lambda_{t+r-1}$ with probability at least 
    $ q = r^{1-\alpha} \cdot \left(16\cdot 48\ln n\right)^{-\alpha}$,
    where
    \[
        \lambda 
        = \theta_{v^\ast}\big(\alpha \log\left({32r}\cdot48\ln n \right)\big)
        +
        \log\left({32r}\cdot48\ln n\right)
        + \log(48\ln n)+3.
    \]
    Suppose now that $v^\ast \in V_t\setminus A_t$, and let 
    \[
        p^\ast 
        = \
        \Pr[v^\ast\in A_{t+1}\cup\dots\cup A_{t+r-1}] - 2^{-|N(v^\ast)\cap A_t|}
        .
    \]
    If $p^\ast \geq \Pr[v^\ast \in B]/2 \geq (96\ln n)^{-1}$, then we can apply \cref{lem:non-active-to-pretty} (using $r-1$ in place of $r$), to obtain that $v^\ast \in A^{\lambda'}_t\cup \dots\cup A^{\lambda'}_{t+r-2}$ with probability at least 
    $q' 
    = 
    r^{1-\alpha} \cdot \left(p^\ast/2\right)^\alpha
    \geq 
    r^{1-\alpha} \cdot \left(192\ln n\right)^{-\alpha}$,
    where
    \[
        \lambda'
        = \theta_{v^\ast}\big(\alpha \log\left({4r}\cdot 96\ln n\right)\big)
        +
        \log\left({4r}\cdot96\ln n\right)
        .
    \]
    If $p^\ast < \Pr[v^\ast \in B]/2$, then we can apply \cref{lem:non-active-to-pretty2}, for $\ell = r - 2 \geq \log(48\ln n) + 4$ and $b_r = \Pr[v^\ast\in B] \geq (48\ln n)^{-1}$, to obtain that that $v^\ast \in A^{\lambda''}_t\cup \dots\cup A^{\lambda''}_{t+r-1}$ with probability at least 
    $q'' 
    = 
    r^{1-\alpha} \cdot \left(64\cdot 48\ln n\right)^{-\alpha}$,
    where
    \[
        \lambda''
        = \theta_{v^\ast}\big(\alpha \log\left({128r}\cdot 48\ln n\right)\big)
        +
        \log\left({128r}\cdot48\ln n\right)
        + \log(4\cdot 48\ln n)+3.
    \]
    In all the settings above, we have $q,q',q''\geq \varepsilon  \ln^{-\alpha} n\cdot (\ln\ln n)^{1-\alpha}$ and $\lambda,\lambda',\lambda'' \leq \beta 
    \ln n\cdot \ln\ln n$, for some constants $\varepsilon,\beta >0$, where the bound on $\lambda,\lambda',\lambda''$ holds because property \cref{def:good-common} 
    and assumption $p \leq c \sqrt{\log(n)/ n}$ imply that for any $v\in V$, $\theta_{v}(i)\leq i\cdot (6c^2+4)\log n$ (recall the definition of $\theta_v$ from \eqref{eq:theta}).
    Therefore, the probability that $v^\ast$ is $(\beta \cdot \ln n\cdot \ln\ln n)$-active at the end of some round in $\{t,\ldots,t+r-1\}$ is at least $\varepsilon \cdot \ln^{-\alpha} n\cdot (\ln\ln n)^{1-\alpha}$, and from \cref{lem:active-vertex-st}, the probability that  $v^\ast \in I_{t + r - 1 +\log n}$ is at least $\varepsilon \cdot \ln^{-\alpha} n\cdot (\ln\ln n)^{1-\alpha} \cdot(2e\beta \cdot \ln n\cdot \ln\ln n)^{-1}$. 
    Thus with at least that probability we have $u \notin V_{t + r - 1 +\log n}$. 
    This completes the proof of \cref{eq:uRstable}.
\end{proof}  



\paragraph{Putting the Pieces Together.}
First, 
\itodo{I modified the proof}
suppose that $p \leq c \sqrt{\log(n)/ n}$ for some constant $c>0$. 
From  \cref{lem:GG-case-large-At,lem:GG-case-large-Vt,lem:GG-case-small-AtVt},
$\Exp[|V_{t+2\log n}|] 
\leq \left(1-\epsilon/\ln^{3.5} n \right) \cdot \Exp[|V_t|]$,
for any $t\geq 0$.
Iteratively applying this inequality, we obtain that for any $i\geq 0$,
\[
    \Exp[|V_{2i\log n}|] 
    \leq 
    \left(1-\epsilon/\ln^{3.5} n \right)^i \cdot n
    .
\]
Substituting $i = 3\ln^{4.5} n/\epsilon$ yields
$
    \Exp\big[|V_{(6/\epsilon)\log n\cdot \ln^{4.5} n}|\big] 
    \leq 
    n^{-2}
    ,
$
and by Markov's inequality, it follows
$
    \Pr[|V_{(6/\epsilon)\log n\cdot \ln^{4.5} n}|\geq 1]
    \leq 
    n^{-2}.
$

If $p \geq \varepsilon/\ln^{2.5} n$
\gtodo{new}
for some constant $\varepsilon>0$, then we use
\cref{lem:GG-case-large-At,lem:GG-case-large-Vt} as above to obtain that
$\Pr[|V_t| \geq  10\ln^2(n)/p] \leq n^{-2}$ for some $t = O(\log^3 n)$.
We also observe that if $|V_t| <  10\ln^2(n)/p$ then the maximum degree of graph $(V,E(V_t))$ is $\Delta < |V_t| \leq  10\ln^2(n)/p \leq 10\ln^{4.5}(n)/\varepsilon$, and \cref{thm:UB-delta}  yields a bound of $O(\Delta\log n) = O(\log^{5.5}n)$.
Combining the two completes the proof of \cref{lem:Good-2s}.


\begin{remark}
    \label{rm:not-tight-analysis}
    Some of the logarithmic factors can be shaved off with a more careful analysis.
    For example, using a ``pipelining'' argument, one could improve the bound on halving $|V_t|$ obtained from \cref{lem:GG-case-small-AtVt}, from $O(\log n \cdot \ln^{3.5}n)$ to $O(\log n + \ln^{3.5}n)$, thus saving one logarithmic factor.
\end{remark}

\section{Logarithmic Switch and the 3-Color MIS Process}
\label{sec:log-switch-3-color}

We present an extension of the 2-state MIS process, called \emph{3-color MIS process}, which uses one additional color, \emph{grey}, and includes also a sub-process, called \emph{logarithmic switch}, which runs in parallel to the main process.
Then we analyze the 3-color MIS process on $G_{n,p}$ random graphs.

\subsection{The Logarithmic Switch Process}
\label{se:log-switch}

We first introduce an abstract \emph{logarithmic switch} process, by specifying its properties.
Then we describe an actual randomized graph process that satisfies these properties with high probability and in a self-stabilizing manner, using $6$ states per vertex. 

\begin{definition}[Logarithmic Switch Process]
    \label{def:switch}
    An \emph{$(a,b)$-logarithmic switch process} on $G=(V,E)$ generates for each vertex $u\in V$ a binary sequence $\sigma_0(u),\sigma_1(u),\ldots,$ where $\sigma_t(u) \in \{\on,\off\}$ for each $t\geq 0$, such that the following properties hold for all $u\in V$.
    \begin{enumerate}[(S1)]
        \item \label{prop:switch-off-ub}
        Every run of consecutive $\off$ values in sequence $\sigma_0(u), \sigma_1(u),\ldots$ has length at most $a\ln n$.    
        
        \item \label{prop:switch-off-lb}
        If $\diam(G) \leq 2$ then every run of consecutive $\off$ values in sequence $\sigma_t(u), \sigma_{t+1}(u),\ldots$ has length at least $\frac a6\ln n$, where $t = \min\{i\geq \frac a6\ln n \colon \sigma_i(u) = \on\}$.    
        
        \item \label{prop:switch-on-ub}
        If $\diam(G) \leq 2$ then every run of consecutive $\on$ values in sequence $\sigma_t(u), \sigma_{t+1}(u),\ldots$ has length at most $b$, where $t$ is some constant independent of $n$.     
    \end{enumerate}
\end{definition}


\begin{definition}[Randomized Logarithmic Switch]
    \label{def:rswitch}
    In the \emph{randomized logarithmic switch process} on $G=(V,E)$,
    each vertex $u\in V$ has a state, called \emph{level}, that takes on values in the set $\{0,1,\ldots,5\}$. 
    The initial value $\level_0(u)$ of $u$ can be arbitrary, and in each round $t \geq 1$ the level of $u$ is updated according to the following rule, which uses a global parameter $0 < \zeta < 1$.
    
    \begin{algorithm}[H] {\small
        \DontPrintSemicolon
        \If{$\level_{t-1}(u) = 5$}
        {choose a random bit $b_t(u)$ such that $\Pr[b_t(u) = 0] = \zeta$\;}
        \uIf{$(\level_{t-1}(u) = 5$ 
        {\bf and} $b_t(u) = 1)$
        {\bf or} $\level_{t-1}(u) = 0$}
            {set $\level_t(u) = 5$\;}
        \lElse
            {set $\level_t(u) = \max\{\level_{t-1}(v)\colon v\in N^+(u)\} -1$}
    }   
    \end{algorithm}
    
    \noindent
    Finally, we define the following mapping of the levels to the binary $\on/\off$ values of \cref{def:switch}.
    For each $u\in V$ and $t\geq 0$,
    \[
        \sigma_t(u) 
        =
        \begin{cases}
            \on & \text{if $\level_t(u) \leq 2$} \\
            \off & \text{if $\level_t(u) \geq 3$}.\\
        \end{cases}
    \]
\end{definition}

\begin{lemma}
    \label{lem:rswitch}
    For any graph $G = (V,E)$, the randomized logarithmic switch process with parameter $0<\zeta\leq1/2$ satisfies properties \cref{prop:switch-off-ub,prop:switch-off-lb,prop:switch-on-ub} for $a = 4/\zeta$ and $b = 3$, with probability $1- O(n^{-2})$, during the first $n$ rounds.    
\end{lemma}

\begin{proof}
    Let $u\in V$, and let $S_v \subseteq V$ be the set of vertices at distance at most $2$ from $u$.
    If $u$ has level at least $3$ in all rounds $t,\ldots, t+a\ln n$, then no vertex $v\in S_u$ has level $0$ in rounds $t+2,\ldots,t+a\ln n$; and at least one vertex $v\in S_u$ must be at level $5$ in all rounds $t+2,\ldots, t+a\ln n-2$.
    It follows that the probability there is some $u\in V$ and $t\leq n$ such that $u$ has level at least $3$ in all rounds $t,\ldots, t+a\ln n$ is at most
    \[
        n^2 (1-\zeta)^{a\ln n - 4} 
        \leq 
        n^{2-a\zeta}/(1-\zeta)^4
        \leq
        16\cdot n^{-2},
    \]
    when $a\zeta = 4$.
    Thus, property \cref{prop:switch-off-ub} holds with probability at least $1-O( n^{-2})$.
    
    Next we assume $\diam(G)\leq 2$.
    The rest of the proof is similar to that in \cite{EmekK21}.
    Observe that there must be a vertex $v$ and a round $t^\ast\leq 5$ such that $\level_{t^\ast}(u) = 5$. 
    And from the end of round $t^\ast+2$, all vertices ``synchronize'' in the sense that once a vertex reaches level $2$ in a round, all vertices reach level $2$ in that round, then the they all reach level $1$ in the next round, then level $0$, and then $5$. 
    It follows that property \cref{prop:switch-on-ub} holds for $b = 3$, starting from round  $t^\ast + 2 \leq 7$.
    The property holds with probability $1$, and for all rounds after round $t^\ast+2$, not just for the first $n$.
    
    As mentioned above, after vertices have synchronized, all $n$ vertices move from level $0$ to level $5$ simultaneously, each time.
    When that happens, the number of rounds until there are no vertices left at level $5$ is greater than $a\ln n-6$ with probability at most
    \[
        n(1-\zeta)^{a\ln n-6} \leq 64\cdot n^{-3}
        ,
    \]
    as before;
    and is smaller than $r = \frac a6\ln n$ with probability at most 
    \[
        (1 - (1-\zeta)^r)^n
        \leq
        e^{-n(1-\zeta)^r}
        \leq
        e^{-n4^{-\zeta r}}
        =
        e^{-n4^{-(a\zeta/6)\ln n }}
        \leq
        e^{-n^{0.07}}
        =
        O(n^{-3})
        .
    \]
    Combining the above, using a union bound, we obtain that property \cref{prop:switch-off-lb} holds with probability $1-O(n^{-2})$.
\end{proof}

\subsection{The 3-Color MIS Process}

We now define the 3-color MIS process, which is an extensions of  the 2-state MIS process.

\begin{definition}[3-Color MIS Process]
    \label{def:3color}
    The process consists of two (sub-)processes that run in parallel on $G=(V,E)$.
    The first is an $(a,3)$-logarithmic switch process, where $a=512$, which generates a value $\sigma_t(u) \in\{\on,\off\}$ for each vertex $u\in V$ in each round $t\geq 0$.
    The second is a variant of the 2-state MIS process, where each vertex $u\in V$ has a state $c_t(u) \in\{\black,\white,\gray\}$, $c_0(u)$ can be arbitrary, and in each round $t\geq 1$, $u$'s state is updated as follows.
    \gtodo{FUTURE-TODO: slightly modify definition so that $\black\to\{\black,\white\}$  when $\sigma_{t-1}(u) = \on$; that way the 2-state process is a special case of the 3-color one; then give only one analysis?}
    
    \begin{algorithm}[H] {\small
        \DontPrintSemicolon
        let $\mathit{NC}_t(u) = \{c_{t-1}(v)\colon v\in N(u)\}$\;
        \uIf{$c_{t-1}(u) = \black$ {\bf and} $\mathit{NC}_t(u)\ni \black$}
        {
            let $c_t(u)$ be a uniformly random state from $\{\black,\gray\}$\;}
        \uElseIf{$c_{t-1}(u) = \white$ {\bf and} $\mathit{NC}_t(u)\not\:\!\ni \black$}
            {let $c_t(u)$ be a uniformly random state from $\{\black,\white\}$\;}
        \uElseIf{$c_{t-1}(u) = \gray$ {\bf and} $\sigma_{t-1}(u) = \on$}
            {set $c_t(u) = \white$\;}
        \lElse
            {set $c_t(u) = c_{t-1}(u)$}
    }    
        \label{alg:3-color-mis}
    \end{algorithm}
    \end{definition}

There are precisely two differences in the update rule above compared to that for the 2-state MIS process: a black vertex with a black neighbor changes to gray with probability $1/2$, rather than to white; and a gray vertex changes to white if its switch value is $\on$.
Note that a gray vertex is treated similarly to a non-active white vertex.  

A vertex is \emph{\textbf{stable}}, if it is black and has no black neighbors, or \emph{it is not black} and has a neighbor that is \emph{stable black}.
Other than that, the remaining definitions and notations are the same as in the 2-state MIS process, 
namely, of \emph{\textbf{active}} vertices, \emph{\textbf{stabilization times}}, $B_t$, $W_t$, $A_t$, $A^k_t$, $I_t$, and $V_t$. 
We also let $\GG_t = V\setminus (B_t\cup W_t)$ denote the set of gray vertices at the end of round $t$.  

The definition of the 3-color MIS process above assumes an arbitrary logarithmic switch process. 
We can use the randomized logarithmic switch from \cref{def:rswitch}, which uses 6 states per vertex, to obtain a 3-color MIS process that uses $6\cdot 3=18$ states in total.
The probability parameter of the randomized switch is $\zeta = 4/a = 2^{7}$, thus at most $7$ random bits are required per round for each vertex (plus one more for each active vertex).

We note that \cref{lem:active-vertex-st,lem:neighbors2-st} and their proofs carry over to the 3-color MIS process, without changes.
We will use also the two simple lemmas below that are specific to the 3-color MIS process.
Recall that $a = 512$ is a parameter of the logarithmic switch.

\begin{lemma}
    \label{lem:gray-active}
    If $t\geq a\ln n$ and $u\in \GG_t$ then $u\in A_{t-a\ln n}\cup\dots\cup A_{t-1}$.  
\end{lemma}
\begin{proof}
    By property \cref{prop:switch-off-ub} in \cref{def:switch} of the logarithmic switch, a vertex is gray for at most $a\ln n$ consecutive rounds. 
    Also if a vertex becomes gray in round $j>0$, it must be active black at the end of round $j-1$.  
    Combining these two facts implies the lemma.
\end{proof}

\begin{lemma}
    \label{lem:freq-black}
    If $\diam(G)\leq 2$, $u\in V$, $t\geq \frac a6 \ln n$, and $t'= t+ \frac a6 \ln n$, then the expected number of times $u$ is active black between rounds $t$ and $t'$ is 
    $
        \Exp\left[\left|\{j\colon u\in B_j\cap V_j\}\cap\{t,\ldots,t'\}\right| \ \middle|\ B_t,W_t \right] 
        \leq 
        4.
    $
\end{lemma}
\begin{proof}
    From properties \cref{prop:switch-off-lb,prop:switch-on-ub} in \cref{def:switch}, it is easy to see that sequence $c_t(u),\ldots,c_{t'}(u)$ contains at most two runs of consecutive black states.
    Moreover, the expected length of the prefix of each black run until $u$ becomes stable black or the run finishes (and $u$ becomes gray) is $2$.
    It follows that $u$ is non-stable black in at most $4$ rounds in expectation.
\end{proof}

\cref{lem:key-lemma,lem:non-active-to-pretty} hold also for the 3-color MIS process, when $u\in V_t\setminus (A_t\cup \GG_t)$ and thus $u\in W_t$.\footnote{The proofs require just minor modifications, mostly replacing some occurrences of ``white"  by ``not black" or ``gray".}
The next simple lemma will be used together with \cref{lem:non-active-to-pretty}.

\begin{lemma}
    \label{lem:gray-happy}
    Let $t\geq 0$, $u\in V$, and $d>0$.
    Let $t'\geq t$ be the first round when either $u$ is white and has at least $d$ black neighbors, or $u$ is stable.
    The expected number of rounds $t < j < t'$ at which $u$ is black and has at least $d$ black neighbors is at most $3$.
\end{lemma}
\begin{proof}
    The lemma is obtained using the observations that: each time $u$'s state changes from white to black, it is equally likely that it remained white;
    and, when $u$ is active black, it becomes gray in the next step with probability $1/2$. 
\end{proof}

\subsection{Analysis of the 3-Color MIS Process on 
 \texorpdfstring{$G_{n,p}$}{Gnp}}
\label{sec:3color-Gnp}

We show that the stabilization time of the 3-color MIS process on $G_{n,p}$ random graphs is $\poly(\log n)$, for the complete range of values of $p$. 

\begin{theorem}
    \label{thm:GnpUB-3s}
    The stabilization time of the 3-color MIS process on a random graph drawn from $G_{n,p}$ is $O(\log^{6} n)$
    with probability $1-O(n^{-2})$.
\end{theorem}

As before, it suffices to show that the above 
bound holds for good graphs, and 
apply  \cref{lem:Gnp-good}.

\begin{lemma}
    \label{lem:Good-2c}
    The stabilization time of the 3-color MIS process on any $(n,p)$-good graph $G=(V,E)$ is $O(\log^{6} n)$
    with probability $1-O(n^{-2})$.
\end{lemma}


\subsubsection{Proof of \texorpdfstring{\cref{lem:Good-2c}}{Lemma Good-2c}}
The proof strategy is similar to \cref{lem:Good-2s}'s: 
From any vector of vertex states at the end of round $t$, we show that the process makes sufficient progress in expectation in $\poly(\log n)$ rounds.
The main difference is that now we show that this is also true even in the case of $|V_t| = O(\log^2(n)/p)$ when $\diam(G) \leq 2$, which corresponds to the case of $p = \Omega (\sqrt{\log(n)/ n})$, by property \cref{def:good-common-lb} in \cref{def:good-graph}.
This is precisely the case that we could not handle in the analysis of the 2 state MIS process.
The relevant lemma is \cref{lem:GG-case-small-AtVt-3c}.

We first observe that \cref{lem:GG-case-large-At}, which considers that case of $|A_t| = \Omega(\log(n)/p)$, carries over to the 3-color MIS process, without any changes in the proof.

Next we consider the case where $|A_t| = O(\ln(n)/p)$, $|V_t| = \Omega(\ln^2 (n)/p)$, and $|\GG_t| = O(\ln^2(n)/p)$.
The following lemma is very similar to \cref{lem:GG-case-large-Vt}, except that it requires also a bound on $|\GG_t|$. 
Recall that $a = 512$ is a parameter of the logarithmic switch.
 
\begin{lemma}\label{lem:GG-case-large-Vt-3c}
  If $|V_t| \geq  82a\ln^2(n)/p$,  
  $|A_t|\leq 80\ln(n)/p$, and
  $|\GG_t|\leq 80a\ln^2(n)/p$,
  then 
  there is a constant $\epsilon>0$ such that
  $
     \Exp[|V_{t+\log n}| ] 
     \leq \left(1-
     \epsilon/\ln n
     \right)
     \cdot |V_t|
     .
  $
\end{lemma}

\begin{proof}
    The proof is very similar to \cref{lem:GG-case-large-Vt}'s.
    As before, from property~\cref{def:good-degree} in \cref{def:good-graph}, the average degree of $G[A_t]$ is at most 
    $
        k = \max\{8p|A_t|,\, 4\ln n\} \leq 640\ln n.
    $
    We let $S$ be a subset of $A_t$ consisting of the  $2|A_t|/3$ vertices $u\in A_t$ with the smallest degree in $G[A_t]$, and let $T = A_t\setminus S$.
    Then for all $u\in S$,
    $|N(u)\cap A_t|\leq 3k$,
     thus $S \subseteq A_t^{3k}$.
%

    The set $V_t$ consist of 
    (i)~all active vertices, $u\in A_t = S\cup T$, 
    (ii)~all non-active non-stable vertices that have some active neighbor, and
    (iii)~all non-active non-stable vertices have no active neighbors (these vertices are gray).
    We can thus partition $V_t$ into the five distinct sets: $S$, $N(S) \setminus N(I_t)$, $T\setminus N(S)$, $N(T) \setminus N^+(S\cup I_t)$, and $V_t\setminus N^+(T\cup Sf)\subseteq \GG_t$.
    We have that $|V_t\setminus N^+(T\cup S)| \leq |\GG_t|\leq 80a\ln^2(n)/p$, 
    $|T\setminus N(S)| \leq |T| < |S|$, and,
    by property~\cref{def:good-STI} in \cref{def:good-graph},
    \[
        |N(T) \setminus N^+(S\cup I_t)|
        \leq
        |N(S) \setminus N(I_t)| + 8\ln^2(n)/p
        .
    \]
    Using these three inequalities, the fact that the sizes of the five sets above sum to $|V_t|$, the assumption $|V_t| \geq  82a\ln^2(n)/p$, and that $a \geq 8$, we obtain
    \[
        |S| + |N(S) \setminus N(I_t)| 
        \geq (|V_t| - (80a+8)\ln^2(n)/p)/2
        \geq (|V_t| - 81a\ln^2(n)/p)/2
        \geq |V_t|/82a
        .
    \]
    Therefore, at least $|V_t|/82a$ vertices $u\in V_t$ are in $S$ or adjacent to $S$. 
    And, from \cref{lem:active-vertex-st},
    each $u\in S \subseteq A_t^{3k}$ is in $I_{t+\log n}$, with probability at least $1/(6ek)$.
    It follows
    \[
        \Exp[|V_{t+\log n}|] 
        \leq 
        |V_t| - (|V_t|/82a) \cdot 1/(6ek)
        \leq
        |V_t| - |V_t|/(1.1\cdot 82a \cdot 10^4\ln n )
        .
        \qedhere
    \]
\end{proof}

Next we assume $|A_t| = O(\ln(n)/p)$ and $|V_t| = \Omega(\ln^2 (n)/p)$, as in the previous lemma, but now $|\GG_t| = \Omega(\ln^2(n)/p)$.
We reduce this case to the previous cases using \cref{lem:gray-active}.

\begin{lemma} \label{lem:GG-case-large-Vt-3c-Gamma}
  If $|V_t| \geq  83a\ln^2(n)/p$,  
  $|A_t|\leq 80\ln(n)/p$, and
  $|\GG_t|> 80a\ln^2(n)/p$,
  then 
  there is a constant $\epsilon>0$ such that
  $
     \Exp[|V_{t + a\ln n +\log n}| ] 
     \leq \left(1-
     \epsilon/\ln n
     \right)
     \cdot |V_t|
     .
  $    
\end{lemma}

\begin{proof}
    Let 
    $
        \tau 
        = 
        \min \{j\geq t\colon |V_j|\leq 82a\ln^2(n)/p
        \text{\; or \;}
        |A_j|\geq 80\ln(n)/p 
        \text{\; or \;}
        |\GG_j|\leq  80a\ln^2 (n)/p
        \}
        .
    $
    We have $\tau \leq t+a\ln n$, because if $|\GG_{t+a\ln n}| > 80a\ln^2 (n)/p$, then \cref{lem:gray-active} implies there is some $j\in \{t,\ldots,t+a\ln n -1\}$ such that $|A_j| \geq |\GG_{t+a\ln n}|/(a\ln n) \geq 80\ln(n)/p$.
    We distinguish three cases depending on which condition in the definition of $\tau$ is satisfied first.
    If $|V_\tau|\leq 82a\ln^2(n)/p$, then 
    \[
        |V_{t+a\ln n}|
        \leq
        |V_\tau|
        \leq 
        82a\ln^2(n)/p 
        \leq
        (1-1/83)\cdot |V_t|
        .
    \]
    If $|A_\tau|\geq 80\ln(n)/p$, then \cref{lem:GG-case-large-At} yields
    $
        \Exp[|V_{t + a\ln n + \log n}| ]
        \leq
        \Exp[|V_{t+\tau + \log n}| ] 
        \leq 
        \left(1-
        \epsilon
        \right)\cdot |V_{t}|
            .
    $
    Last, if $|\GG_\tau| \leq  80a\ln^2 (n)/p$ and the other two conditions do not hold, then \cref{lem:GG-case-large-Vt-3c} gives
    $\Exp[|V_{t + a\ln n +\log n}| ] 
     \leq \left(1-
     \epsilon/\ln n
     \right)
     \cdot |V_t|.
$
\end{proof}

The next two lemmas deal with the case of $|V_t| = O(\ln^2(n)/p)$. 
The first one assumes $\diam(G)\leq 2$, and thus covers the case of $p = \Omega (\sqrt{\log(n)/ n})$, by property \cref{def:good-common-lb} in \cref{def:good-graph}; while the second lemma assumes $p = O (\sqrt{\log(n)/ n})$ and is similar to \cref{lem:GG-case-small-AtVt}.

\begin{lemma}
  \label{lem:GG-case-small-AtVt-3c}
  For any $t\geq \frac a6 \ln n$, if $|V_t| \leq  83a \ln^2(n)/p$ and $\diam(G) \leq 2$ then 
  there is a constant $\epsilon > 0$ such that
  $
     \Exp[|V_{t + \frac76 a\log n + \log n}| ] 
     \leq \left(1-
     \epsilon/\ln^{3} n
     \right)
     \cdot |V_t|
     .
  $
\end{lemma}

\begin{proof}
    From property~\cref{def:good-degree} in \cref{def:good-graph}, the average degree of induced subgraph $G[V_t]$ is at most 
    $
        k = \max\{8p|V_t|,\, 4\ln n\} \leq 664a\ln^2 n.
    $
    Let $T$ be a subset of $V_t$ consisting of the  $\min\{\ln(n)/p,\, |V_t|/2\}$ vertices $u\in V_t$ with the largest degree in $G[V_t]$, and let $S = V_t\setminus T$.
    Then $|S|\geq |T|$, and all $u\in S$, $|N(u)\cap V_t|$ is at most
    \[
        d =  k |V_t|/|T| \leq k \cdot \max\{p|V_t|/\ln n,\, 2\}
        \leq
        55112\ln^3 n
        .
    \]
    From property \cref{def:good-ST} in \cref{def:good-graph}, 
    $|E(S,T)| \leq 6|S|\ln n$.
    Let $R =  \{u\in S\colon |N(u)\cap T| \leq 12\ln n \}$.
    Then $|R| \geq |S|/2 \geq |V_t|/4$.
    We will show that, for some constant $\epsilon' > 0$,
    \begin{equation}
        \label{eq:uRstable-3c}
        \Pr[u\notin V_{t+\frac76 a\ln n + \log n}] 
        \geq 
        \epsilon' \ln^{-3}n,\
        \text{ for all }
        u\in R.
    \end{equation}
    From this, it follows that
    $
        \Exp[|V_{t+\frac76 a\ln n + \log n}| ] 
        \leq 
        |V_t| - (|V_t|/4) \cdot \epsilon' \ln^{-3} n
        .
    $    
    To complete the proof of the lemma it remains to prove \cref{eq:uRstable-3c}.  
    
    Let $u\in R$, and suppose that $u\notin \GG_t$ (we deal with the case $u\in \GG_t$ at the end).
    From \cref{lem:freq-black}, the expected value of $\sum_{t\leq j\leq t+\frac a6\ln n}|(N(u)\cap T)\cap(B_j\cap V_j)|$, that is, the total number of times that vertices $v\in N(u)\cap T$ are active black between rounds $t$ and $\frac a6\ln n$, is at most $4\cdot |N(u)\cap T| \leq 4\cdot12\ln n$. 
    Then, by Markov's inequality, that number is at most $5\cdot 12\ln n$ with probability at least $1/5$. 
    And since $\frac a6 > 5\cdot 12$, it follows that, with probability at least $1/5$, there is some $j\in\{t,\ldots,t+\frac a6\ln n\}$ such that $(N(u)\cap T)\cap(B_{j}\cap V_{j}) = \emptyset$.
    
    Next we claim that, if  $(N(u)\cap T)\cap(B_{j}\cap V_{j}) = \emptyset$ for some $j\geq t$, then (i)~$u\notin V_{j}$, or (ii)~$u\in A_{j'}$ for some $t\leq j'< j$, or (iii)~$(N^+(u)\cap S)\cap A_{j} \neq \emptyset$. 
    Indeed, suppose that (i) and (ii) do not hold, i.e., $u\in V_{j}$ and $u\notin A_{t}\cup\dots\cup A_{j-1}$.
    From $u\in V_{j}$, it follows $N^+(u)\cap I_{j} = \emptyset$.
    From $u\notin A_{t}\cup\dots\cup A_{j-1}$ and the assumption $u\notin\GG_t$, it follows $u \in W_{j}$.
    Then, if $(N(u) \cap S)\cap B_{j} \neq \emptyset$, each vertex $v\in (N(u) \cap S)\cap B_{j}$ is in $A_{j}$; while if $(N(u)\cap S) \cap B_{j} =\emptyset$, then $N(u) \cap B_{j} =\emptyset$ and $u\in A_{j}$.  
    Therefore (iii) holds.

    From the above, it follows that with probability at least $1/5$, there is some $t\leq j\leq t+ \frac a6\ln n$ such that $u\notin V_{j}$ or $(N^+(u)\cap S)\cap A_{j} \neq \emptyset$.
    And if $v\in (N^+(u)\cap S)\cap A_{j}$, then $v\in A_j^d$, and from \cref{lem:active-vertex-st}, the probability that
    $v\in I_{j+\log n}$ is at least $1/(6ed)$.
    We conclude that
    \[
        \Pr[u\notin V_{t+\frac a6\ln n + \log n}]
        \geq 
        (1/5)
        \cdot
        1/(6ed)
        \geq
        (4.5\cdot10^6\ln^3 n)^{-1},
    \]
    which implies \cref{eq:uRstable-3c}.
 
    Finally, if $u\notin \GG_t$, we consider the first round $j>t$ such that $u\notin \GG_j$.
    From property \cref{prop:switch-off-ub}, 
    $j\leq t+ a\ln n$.
    Then we apply the result for the previous case to complete the proof of \cref{eq:uRstable-3c}.
\end{proof}


\begin{lemma}\label{lem:GG-case-small-AtVt-3c-small-p}
  If $|V_t| \leq  83a\ln^2(n)/p$ and $p \leq c \sqrt{\log(n)/ n}$ for some constant $c>0$, then there is a constant $\epsilon = \epsilon(c) > 0$ such that
  $
     \Exp[|V_{t + \log^{1.1} n}| ] 
     \leq \left(1-
     \epsilon/\ln^{3.9} n
     \right)
     \cdot |V_t|
     . 
  $
\end{lemma}
\begin{proof}[Proof Sketch]
    We define the set $S,T,R$ and the degree thresholds $k,d$ as in the proof of \cref{lem:GG-case-small-AtVt-3c}, and we show
    \begin{equation}
        \label{eq:uRstable-3c-small-p}
        \Pr[u\notin V_{t+ \log^{1.1} n}] 
        \geq 
        \epsilon' \ln^{-3.9}n,
        \
        \text{ for all }
        u\in R,
    \end{equation}
    which implies the lemma.
    Next we prove \cref{eq:uRstable-3c-small-p}.  
    
    Let $u\in R$, and suppose that $u\notin \GG_t$ (we deal with case $u\in \GG_t$ at the end).
    For each $v\in N(u)\cap T$, let $t_v \geq t$ be the first round when either $v$ is white and has at least $\ell = \ln n$ black neighbors, or is stable; and let $x_v$ be the number of rounds $t \leq j \leq \min\{t_v, t + r\}$ at which $v$ is black and has at least $\ell$ black neighbors, where $r = 12\ln n\cdot \ln^2\ln n$.
    From  \cref{lem:gray-happy}, the probability that $x_v \leq \ln^2\ln n$ for all $v$ is at least 
    $1 - |N(u)\cap T|\cdot e^{-\Omega(\ln^2\ln n)} = 1 - e^{-\omega(\ln\ln n)}$. 
    For each $v\in N(u)\cap T$ let $p_v$ be the conditional probability that $v\in B_{t_v+1}\cup\dots\cup B_{t + r}$, given $B_t,W_t$.

    If $\sum_{v\in N(u)\cap T} p_v \leq 1/2$, then with probability at least $1/2 - e^{-\omega(\ln\ln n)} > 1/3$, the total number of rounds in which at least one $v\in N(u)\cap T$ is black and has at least $\ell$ black neighbors is at most 
    $|N(u)\cap T| \cdot \ln^2\ln n \leq 12\ln n\cdot \ln^2\ln n \leq r$, thus there is some $j\in\{t,\ldots,t+r\}$ such that no $v\in N(u)\cap T$ is black and has at least $\ell$ black neighbors.
    Then we can infer that with probability at least $1/3$ some vertex in $N^+(u)$ is stable black or is $d$-active at some round in $\{t,\ldots,t+r\}$, 
    in the same way as in the proof of \cref{lem:GG-case-small-AtVt-3c},  and then obtain \cref{eq:uRstable-3c-small-p} using \cref{lem:active-vertex-st}.  

    If $\sum_{v\in N(u)\cap T} p_v > 1/2$, then there is some $v^\ast \in N(u)\cap T$ such that $p_{v^\ast} \geq (2|N(u)\cap T|)^{-1} \geq (24\ln n)^{-1}$.
    We can then apply \cref{lem:non-active-to-pretty} to $v^\ast$ at round $t_{v^\ast}$ to show that the probability vertex $v^\ast$ is $z$-active at some round in $\{t,\ldots,t+r\}$, where $z = \theta_u\Big(\alpha\log\frac{4r}{p_{v^\ast}-2^{-\ell}}\Big) + \log \frac{4r}{p_{v^\ast}-2^{-\ell}} = O(\log n \cdot \log\log n)$, is at least $r^{1-\alpha}\cdot \left(\frac{p_{v^\ast}-2^{-\ell}}{2}\right)^\alpha = \Omega(\ln^{3.9} n)$, as $\alpha \leq 2.41$
    Again we obtain \cref{eq:uRstable-3c-small-p} using \cref{lem:active-vertex-st}.

    Finally, as before, if $u\notin \GG_t$, we consider the first round $j>t$ such that $u\notin \GG_j$, and apply the result for the previous case to complete the proof of \cref{eq:uRstable-3c-small-p}.
\end{proof}


We can now conclude the proof of \cref{lem:Good-2c}, as we did for \cref{lem:Good-2s}. 
From  \cref{lem:GG-case-large-At,lem:GG-case-large-Vt-3c,lem:GG-case-large-Vt-3c-Gamma,lem:GG-case-small-AtVt-3c,lem:GG-case-small-AtVt-3c-small-p}, we have that for any $t\geq \frac a6\ln n$, 
$\Exp[|V_{t+\log^{1.1} n}|] 
\leq \left(1-\epsilon/\ln^{3.9} n \right) \cdot \Exp[|V_t|].$
Iteratively applying this inequality, and using by Markov's inequality, we  obtain as before  
$
    \Pr[|V_{c'\ln^{6} n}|\geq 1]
    \leq 
    n^{-2},
$
for a large enough constant $c'> 0$.
This completes the proof of \cref{lem:Good-2c}.

\appendix 

\section*{APPENDIX}

\section{Omitted Proofs}

\subsection{Proof of \texorpdfstring{\cref{lem:neighbors2-st}}{Lemma Veighbors2}}

\label{sec:proof-neighbors2-st}

  We assume $k_1\leq k_2\leq\cdots\leq k_\ell$. 
  For $1\leq i\leq\ell$,
  let $r_i = \lceil\log(k_i+1)\rceil$, 
  let $\EE_{i}$ be the event that 
  $\phi_{t+1}(u_i)=\dots=\phi_{t+r_i}(u_i)=\black$,
  and let $\EE = \bigcup_i\EE_{i}$.
  Then
  \[
    \Pr[\EE] 
    = 
    1 - \prod_i\left(1-\frac1{2^{r_i}}\right)
    \geq
    1 - \prod_i\left(1-\frac{1}{2k_i}\right)
    \geq
    \left(1-e^{-\sum_i \frac1{2k_i}}\right)
    \geq 
    \left(1-e^{-1}\right)\cdot
    \min\left\{1,\, \sum_i \frac1{2k_i}\right\}
    .
  \]
  Suppose that $\EE$ occurs and let $j$ be the smallest index such that $\EE_{j}$ occurs, i.e., $\bar\EE_{1}\cap\cdots\cap\bar\EE_{j-1}\cap\EE_{j}$ occurs.
  If $g_j = |N(u_j)\cap\{u_1,\ldots,u_{j-1}\}|$,
  then the probability that none of the $k_j$ vertices $v\in N(u_j)\cap A_t$ satisfies $\phi_{t+1}(v)=\dots=\phi_{t+r_j}(v)=\black$ is 
  \[  
    (1-2^{-r_j})^{k_j - g_j}
    \geq
    (1-2^{-r_j})^{k_j}
    \geq
    e^{-1},
  \]
  similarly to \cref{eq:coins_2}.
  Combining this with the previous inequality we obtain that the probability that $u_i \in I_{t+r_i}$ for at least one vertex $u_i\in \{u_1,\ldots,u_\ell\}$ is at least 
  $
    e^{-1}\cdot \left(1-e^{-1}\right)\cdot \min\big\{1,\, \sum_i \frac1{2k_i}\big\}
    \geq 
    \frac15 \cdot \min\big\{1,\, \sum_i \frac1{2k_i}\big\}
    .
  $





\subsection{Proof of \texorpdfstring{\cref{lem:active-to-pretty}}{Lemma Act-to-kAct}} 

\label{sec:proof-of-active-to-pretty}

    Let $\BB$ be the event $u\in B_{t+\ell}\cup\dots\cup B_{t+r}$; then $\Pr[\BB] = b_r$. 
    Let 
    \[
        \tau = \min \{j > t \colon u \in W_j \text{ or } |N(u) \cap B_j|\leq k \}
    \]
    be the first round $j>t$ at the end of which $u$ is white or has at most $k$ black neighbors.
    We have $\Pr[\tau > t+ \ell] \leq 2^{\ell} \leq  b_r/4$, since $\tau > t + \ell$ implies $\phi_{t+1}(u)=\dots=\phi_{t+\ell}(u)=\black$.
    Thus $\Pr[\tau\leq t+\ell]\geq 1 - b_r/4$. 
    Let 
    \[
        x = \Pr[|N(u) \cap B_\tau| \leq k \mid \tau\leq t+\ell]
        .
    \]
    We distinguish two cases, $x\geq b_r/4$ and $x \leq b_r/4$.

    First suppose that $x \geq  b_r/4$. 
    For any given $j>t$, 
    \[
        \Pr[u\in A_j \mid \tau = j,\, |N(u) \cap B_\tau| \leq k] = 1/2.
    \]
    The reason is that $u\in B_{j-1}$ and $|N(u) \cap B_{j-1}| > k>0$ if $\tau=j>t+1$, and $u \in A_t = A_{j-1}$ if $\tau = j=t+1$.
    In either case $u\in A_{j-1}$, thus the state of $u$ at the end of round $j$ is chosen uniformly at random, independently of the remaining choices in round $j$.
    In particular, $u$ is black with probability $1/2$ when $0 < |N(u) \cap B_{j}|\leq k$, and is white with probability $1/2$ when $|N(u) \cap B_{j}| = 0$.
    It follows that
    \[
        \Pr[\{u\in A_\tau\} \cap \{|N(u) \cap B_\tau| \leq k\} \cap \{\tau\leq t+\ell\}]\geq (1/2)\cdot x \cdot (1-b_r/4)
        \geq 
        3 b_r/32.
    \]
    Since the event on the left side implies that $u$ is $k$-active at the end of round $\tau \leq t+\ell< t+r$, and $3 b_r/32$ is greater than the desired lower bound for $q_r$, the lemma holds in this case.

    Suppose now that $x \leq b_r/4$. 
    Then
    \[
        \Pr[\BB \cap \{|N(u) \cap B_\tau| > k\}\cap \{\tau\leq t+\ell\}]
        \geq
        \Pr[\BB]
        -
        \Pr[\tau > t+\ell]
        -
        x 
        \geq 
        b_r - b_r/4 - 
        b_r/4 
        = 
        b_r/2
        .
    \]
    If $\tau \leq t+\ell$ and $|N(u) \cap B_\tau| > k$ (and thus $u\in W_\tau$ by $\tau$'s definition), we define the following events:
    $\AA^k$ is the event that $u \in A^k_{\tau+1}\cup\dots \cup A^k_{t+r-1}$;
    $\AA$ is the event that $u \in A_{\tau+1}\cup\dots \cup A_{t+r-1}$; and
    $\XX$ is the event that the states of vertices at the end of round $\tau$ are such that the conditional probability of $\AA$, given these states and $\tau$, is at least $b_r/4$.
    
    If $\tau \leq t+\ell$ and $|N(u) \cap B_\tau| > k$, then event $\BB$ implies $\AA$, because vertex $u$, which is non-active white at the end of round $\tau$, cannot become black before becoming active first.
    Thus, from the last inequality above, it follows
    \[
        \Pr[\AA \cap \{|N(u) \cap B_\tau| > k\}\cap \{\tau\leq t+\ell\}]
        \geq
        b_r/2
        .
    \]
    Also
    \[
        \Pr[\AA \cap \XX \cap \{|N(u) \cap B_\tau| > k\}\cap \{\tau\leq t+\ell\}]
        \geq
        b_r/2
        -
        b_r/4
        =
        b_r/4
        .
    \]
    We can now apply \cref{lem:non-active-to-pretty},
    starting from  round $\tau\leq t+\ell$, using $d>k\geq \log(1/b_r) +3$ and $p_r\geq b_r/4$,
    to obtain
    \[
        \Pr[\AA^k \cap \XX \cap \{|N(u) \cap B_\tau| > k\}\cap \{\tau\leq t+\ell\}]
        \geq
        r^{1-\alpha} \cdot \left(\frac{b_r/4-2^{k}}{2}\right)^\alpha
        \geq
        r^{1-\alpha} \cdot \left(\frac{b_r/4-b_r/8}{2}\right)^\alpha
        .
    \]
    It follows that $q_{r} = \Pr[\AA^k]\geq r^{1-\alpha} \cdot \left(\frac{b_r/4-b_r/8}{2}\right)^\alpha$, which concludes the proof of this case.  

\subsection{Proof of \texorpdfstring{\cref{lem:non-active-to-pretty2}}{Lemma NAct-to-kAct2}}

\label{sec:proof-non-active-to-pretty2}

\begin{proof}
    We have $\Pr[u\in A_{t+1}] = 2^{-d}$ and $\Pr[u\in (A_{t+2}\cup\dots\cup A_{t+r-1})\setminus A_{t+1}] = p_r - 2^{-d}$.
    We also note that if $u\notin A_{t+1}$ then $u\in W_{t+1}$, and $u$ may become black in a subsequent round only after it becomes active. 
    It follows that
    \begin{align*}        
        \Pr[
        \{u\in (B_{t+\ell}\cup\dots\cup B_{t+r})\cap A_{t+1} 
        ]
        &=
        b_r
        -
        \Pr[
        u\in (B_{t+\ell}\cup\dots\cup B_{t+r})\setminus A_{t+1} 
        ]
        \\&
        \geq
        b_r
        -
        \Pr[u\in (A_{t+2}\cup\dots\cup A_{t+r-1})\setminus A_{t+1}]
        \\&
        \geq
        b_r - (p_r - 2^{-d})
        \\&
        \geq
        b_r/2
        .
    \end{align*}
    Let $\XX$ be the event that the 
    states of vertices
    at the end of round $t+1$ are such that the conditional probability of $u \in B_{t+\ell}\cup\dots\cup B_{t+r}$ 
    is at least $b_r/4$.
    Then 
    \begin{align*}        
        \Pr[
        \{u \in (B_{t+\ell}\cup\dots\cup B_{t+r})\cap A_{t+1}\} 
        \cap 
        \XX
        ]
        &\geq 
        b_r/2 
        - 
        \Pr[
        \{u \in (B_{t+\ell}\cup\cdots \cup B_{t+r})\cap A_{t+1}\} 
        \cap 
        \bar\XX
        ]
        \\&
        \geq 
        b_r/4
        .
    \end{align*}
    We can now apply \cref{lem:active-to-pretty}, starting from round $t+1$ and using $b_r/4$ in place of $b_r$, to obtain 
    \[
        \Pr[
        \{u \in (A^k_{t+1}\cup\dots\cup A^k_{t+r-1})\cap A_{t+1}\} 
        \cap 
        \XX
        ]
        \geq 
        r^{1-\alpha} \cdot \left({b_r}/{64}\right)^\alpha
        .
    \]
    This implies the lemma.
\end{proof}

\subsection{Proof of \texorpdfstring{\cref{lem:Gnp-good}}{Lemma Gnp-is-Good}}

\label{sec:proof-Gnp-good}

The proof of 
consists of a series of lemmas.
In all these lemmas, the graph $G=(V,E)$ considered is a random graph drawn from $G_{n,p}$.



Property \cref{def:good-degree} holds trivially for sets $S$ of size $k\leq 4\ln n$.
The next lemma (applied for all $k > 4\ln n$, and then combining the results using a union bound) shows that $G$ satisfies the property for all larger sets, with probability at least $1-n^{-\Omega(\log n)}$.

\begin{lemma}\label{lem:Gnp-subset-degrees}
  Let $G=(V,E)$ be a random graph drawn from $G_{n,p}$, and let $k\geq 1$. 
  With probability at least  $1-n^{-k}$, all subgraphs of $G$ on $k$ vertices have at most $\max\{4pk^2,\, 2k\ln n\}$ edges.
\end{lemma}
\begin{proof}
  The probability there is a subgraph with $k$ vertices and at least $r = \max\{2k\ln n,\, 4pk^2\}$ edges is at most 
  \[
    \binom{n}{k}\cdot \binom{k^2/2}{r}\cdot p^r
    \leq
    n^{k} \cdot
    \left(\frac{ek^2}{2r}\right)^r \cdot p^r
    =
    e^{k\ln n - r\ln\frac{2r}{epk^2}}
    \leq 
    e^{k\ln n - 2k\ln n \cdot \ln\frac{8pk^2}{epk^2}}
    \leq 
    n^{- k}
    .\qedhere
  \]
\end{proof}

The next lemma shows that $G$ satisfies property~\cref{def:good-subset-neighbors} with probability $1-n^{-\Omega(\log n/p)}$

\begin{lemma}\label{lem:Gnp-subset-neighbors}
  Let $G=(V,E)$ be a random graph drawn from $G_{n,p}$, and let $k \geq 40\ln(n)/p$.
  With probability at least  $1-n^{- k}$, every set $S\subseteq V$ of size $|S| = k$ satisfies 
  \[
    |\{u\in V \colon |N(u)\cap S)| < pk/2 \}|
    \leq k/2
    .
  \]
\end{lemma}
\begin{proof}
  For any set $S$ of size $k$, and any vertex $u\in V\setminus S$, the expected number of neighbors of $u$ in $S$ is $pk$.
  By a Chernoff bound, the probability that $u$ has fewer than $pk/2$ neighbors in $S$ is at most 
  $
    e^{-pk/8}.
  $
  Then the probability there is some set $S$ of size $k$ such that at least $k/2$ vertices $u\in V\setminus S$ have fewer than $pk/2$ neighbors in $S$, is at most
  \[
    \binom{n}{k}\cdot \binom{n-k}{k/2}\cdot e^{-(k/2)\cdot pk/8}
    \leq
    n^{k}
    \cdot
    n^{k/2}
    \cdot e^{-pk^2/16}
    =
    e^{(3/2)k\ln n - pk^2/16}
    \leq
    n^{- k}
    .
    \qedhere
  \]
\end{proof}

\begin{lemma}\label{lem:Gnp-neighborhood}
  Let $G=(V,E)$ be a random graph drawn from $G_{n,p}$, and let $k = 3\ln(n)/p$.
  With probability at least $1-n^{-k}$, every set $S\subseteq V$ of size $|S| \geq  k$ satisfies $|V\setminus N^+(S)| \leq k$.
\end{lemma}

\begin{proof}
  The probability there is a set $S$ of size $k$ with $|V\setminus N^+(S)|  \geq k$ is at most
  \[
    \binom{n}{k}\cdot \binom{n-k}{k}\cdot (1-p)^{k^2}
    \leq
    n^k
    \cdot
    n^k
    \cdot e^{-pk^2}
    =
    e^{2k\ln n -pk^2}
    =
    n^{- k}
    .
    \qedhere
  \]
\end{proof}

The next lemma shows that $G$ satisfies property~\cref{def:good-STI} with probability $1-n^{-\Omega(\log n/p)}$.

\begin{lemma}\label{lem:Gnp-STI}
  Let $G=(V,E)$ be a random graph drawn from $G_{n,p}$.
  With probability at least  $1 - n^{-\ln(n)/p}$,
  every triplet of disjoint sets $S,T,I \subseteq V$, such that $|S|\geq 2|T|$ and $(S\cup T) \cap N(I)= \emptyset$, satisfies
  \begin{equation}
    \label{eq:Gnp-STI}
    |N(T) \setminus N^+(S\cup I)| - |N(S)\setminus N^+(I))| 
    \leq 
    8\ln^2(n)/p 
    .  
  \end{equation}
\end{lemma}
\begin{proof}
  From \cref{lem:Gnp-neighborhood}, with probability at least $1-n^{-3\ln(n)/p}$, all sets $S,I\subseteq V$ such that $|S\cup I| \geq 3\ln(n)/p$ satisfy
  $
    |V\setminus N^+(S\cup I)| \leq 3\ln(n)/p,
  $ 
  and thus
  \[
    |N(T) \setminus N^+(S\cup I)|
    \leq 
    |V \setminus N^+(S\cup I)|
    \leq 
    3\ln(n)/p
    ,
  \]
  which implies \cref{eq:Gnp-STI}.
  
  Next we assume that $|S\cup I| \leq 3\ln(n)/p$.
  Since $|S| \geq 2|S|$, we have
  $|S\cup T \cup I| \leq 4.5\ln(n)/p$, thus there are at most $n^{4.5\ln(n)/p}$ different triplets $S,T,I$.
  Choose one such triplet $S,T,I$, before revealing the edges of $G$. 
  Then reveal the edges incident to vertices $u\in I$; this determines $N(I)$.
  Let $U = V \setminus (S\cup T \cup N^+(I))$.
  The two sets on the left side of \cref{eq:Gnp-STI} can then be expressed as $N(T) \setminus N^+(S\cup I) = U \cap N(T) \setminus N(S)$, and $N(S) \setminus N^+(I) = U \cap N(S)$.
  For every $u\in U$, the probability that $u\in N(T)\setminus N(S)$ is 
  \[
    p_1 
    = 
    \Pr[u \in N(T) \setminus N(S)]
    =
    \big(1 - (1-p)^{|T|}\big)\cdot (1-p)^{|S|}
    ,
  \]
  and the probability that $u\in N(S)$ is
  \[
    p_2 
    = 
    \Pr[u\in N(S)]
    =
    1 - (1-p)^{|S|}
    .
  \]
  Letting $\varepsilon = (1-p)^{|T|}$ and using that $|S|\geq 2|T|$, we obtain $p_1 \leq (1-\varepsilon) \cdot \varepsilon^2$ and $p_2 \geq 1 - \varepsilon^2$.
  Thus
  \[
    \frac{p_2}{p_1}
    \geq
    \frac{1 - \varepsilon^2}{(1-\varepsilon)\cdot \varepsilon^2} 
    =
    \frac{1 + \varepsilon}{\varepsilon^2}
    \geq
    2
    .
  \]
  It follows that, by considering all vertices $u\in U$ one after the other, and revealing all edges incident to each $u$ at the moment $u$ is considered, we can analyze the difference $$D = |U \cap N(T) \setminus N(S)| - |U \cap N(S)| = |N(T) \setminus N^+(S\cup I)| - |N(S)\setminus N^+(I)|$$ as a biased random walk on the integers starting at $0$, and moving to the right with probability $p_1$ and to the left with probability $p_2$.
  The probability that the (infinite) random walk every reaches value $i\geq 1$ is know to be $\left({p_1}/{p_2}\right)^i\leq 2^{-i}$.
  Thus, 
  $
    \Pr[D \geq i] \leq 2^{-i}
    .
  $
  And the probability that $D \geq 
  {8\ln^2(n)}/{p}$ for at least one possible triplet $S,T,I$ is then at most
  \[
    n^{4.5\ln(n)/p}\cdot 2^{-{8\ln^2(n)}/{p}}
    \leq
        n^{-{\ln n}/{p}}
        .
        \qedhere
  \]
\end{proof}

Property \cref{def:good-ST} holds trivially for sets $S$ of size $k \leq 6\ln n$, since $|S|\geq |T|$.
The next lemma (applied for all $k > 6\ln n$) shows that $G$ satisfies the property with probability at least $1-n^{-\Omega(\log n)}$ for all larger sets.

\begin{lemma}\label{lem:Gnp-ST}
  Let $G=(V,E)$ be a random graph drawn from $G_{n,p}$, and let $k\geq 1$.
  With probability at least  $1 - n^{-2k}$,
  every pair of disjoint sets $S,T \subseteq V$, such that $|S| = k \geq |T|$ and $|T| \leq \ln(n)/p$, satisfies
  $|E(S,T)| \leq 6k\ln n$.
\end{lemma}
\begin{proof}
  For any given pair $S,T$, the expected value of $|E(S,T)|$ is $p\cdot|S|\cdot|T| \leq k\ln n$, and by a Chernoff bound, the probability that $|E(S,T)| \geq 6k\ln n$ is at most $2^{-6k\ln n}$.
  Then the probability there is at least one pair $S,T$ such that $|E(S,T)| \geq 6k\ln n$ is at most
  \[
    n^{|S|}\cdot n^{|T|}\cdot 2^{-6k\ln n}
    \leq
    n^{2k}\cdot 2^{-6k\ln n}
    \leq
    n^{-2k}
    .
    \qedhere
  \]
\end{proof}

Our last lemma implies that properties \cref{def:good-common,def:good-common-lb} hold with probability $1-O(n^{-2})$.

\begin{lemma}\label{lem:Gnp-common-neighbors}
  In a random graph $G$ drawn from $G_{n,p}$, the probability that no two vertices have $k$ common neighbors is at least 
  $1 - n^2\cdot ({ep^2n}/{k})^k$.
  And the probability that $\diam(G) \leq 2$ is at least $1-n^2\cdot e^{-p^2(n-1)}$.
\end{lemma}
\begin{proof}
  The probability there is a pair of vertices that have at least $k$ common neighbors is at most
  $
    \binom{n}{2}\cdot \binom{n-2}{k}\cdot p^{2k}
    \leq
    n^2\cdot \left(\frac{ep^2n}{k}\right)^k.
  $
  And the probability there is a pair of vertices with no common neighbors and no adjacent to each other is
  $\binom{n}{2}\cdot 
  (1-p)\cdot (1-p^2)^{n-2}
  \leq
  n^2\cdot e^{-p^2(n-1)}$.
\end{proof}






\section{Other Related Work}
\label{sec:related-work}

In 1985, Luby \cite{Luby86} proposed a simple distributed randomized algorithm that finds an MIS in time $O(\log n)$ w.h.p. 
Simultaneously, Alon et al.\ \cite{AlonBI86} proposed a similar algorithm with the same performance. 
Both algorithms work with $O(\log n)$-bit messages and need access to $O(\log n)$ random bits at each round. 



Due to various applications in radio sensor networks, restricted distributed models of communication were introduced, in which the MIS problem has been widely studied. 
In the beeping model, introduced  by Cornejo and Kuhn \cite{CornejoK10}, nodes have no knowledge of the local or global structure of the network, do not have access to synchronized clocks and the communication among nodes relies completely on carrier sensing (as described in the introduction). 
Afek et al.~\cite{AfekABCHK13} show that
in the version of the beeping model where nodes are initially asleep and are woken up by an adversary, it is not possible to locally converge to an MIS in sub-polynomial time. 
Therefore, they consider various relaxations on the model, providing algorithms converging to an MIS in a polylogarithmic number of rounds. 
In detail, if the nodes know an upper bound on the size of the network, or if the beeping nodes are awakened by the neighbor's beep, the MIS can be found in time $O(\log^3 n)$ w.h.p. 
If the nodes have synchronous clocks, an MIS can be found in time $O(\log^2 n)$ w.h.p. 
We remark that the authors provide a self-stabilizing algorithm just in the first setting, i.e. when an upper bound on the size of the network is known by the nodes and that. 
In all algorithms, the nodes have super-constnt state and have access to a super-constant number of random bits.

In the version of the beeping model with
synchronized clocks, collision detection, and simultaneous wakeup, 
Afek et al.\ \cite{AfekABHBB11} had earlier shown that the MIS problem is solved by a biological process in time $O(\log^2 n)$ w.h.p.
\cite{AfekABCHK13} showed that this bound is also achievable without knowledge of an upper bound on the size of the network. 
Jeavons et al.\ \cite{JeavonsS016} improved these results, showing that an MIS can be found in time $O(\log n)$ w.h.p. 
An improved analyisis of the local complexity of this algorithm was provided by Ghaffari~\cite{Ghaffari17}. 
In the same version of the beeping model without collision detection, Holzer and Lynch in \cite{HolzerL17}, proposed a variant of the algorithm of \cite{Ghaffari16}, and showed that it converges locally in time $O((\log\Delta + \log 1/\varepsilon)\cdot \log 1/\varepsilon)$ with probability at least $1-\varepsilon$ on a network with maximum degree $\Delta$. 
All these algorithms require super constant space and random bits per round.

Emek et al.~\cite{EmekW13} introduced the stone age model, \gtodo{I do not understand this sentence} 
inspired by biological cellular networks or networks of microprocessor devices. 
\gtodo{\cite{EmekW13} also gives an MIS algorithm; we must mention that}
In this model, the nodes can communicate by transmitting messages belonging to a finite communication alphabet. The nodes communicate in an asynchronous environment, where the pattern is decided by an adversary, and they have no knowledge about the size of the network. In the stone age model, the MIS problem was considered by \cite{EmekW13,EmekK21}. In~\cite{EmekW13}, is provided an algorithm that compute a MIS in $O(\log^2 n)$ rounds. However, it assumes that all the nodes have the same initial state, and therefore is not self-stabilizing. In \cite{EmekK21}, they provided a self-stabilizing algorithm that stabilizes in time $O((D + \log n) \log n)$ w.h.p., and the possible number of states of each node is $O(D)$, where $D$ is the diameter of the graph. 


In \cite{MetivierRSZ11}, the authors introduced a randomized distributed algorithm that finds an MIS in time $O(\log n)$ w.h.p. In particular, the algorithm is an adaptation of Luby's algorithm so that messages of just 1 bit are used. They consider an anonymous network, but in their setting, the vertices can distinguish between their neighbors, and each vertex needs a number of states that depends on $n$ and the node degree.

MIS algorithm has also received a lot of attention from the Self-Stabilization community. 
For a survey of those algorithms see \cite{GuellatiK10}.

We first cite here the self-stabilizing algorithm for non-anonymous networks, i.e. where vertices have IDs. \gtodo{mention all these are deterministic} 
In \cite{GoddardHJS03}, the authors provide a simple deterministic distributed algorithm that stabilizes on an MIS in $O(n)$ time and $O(n^2)$ moves (i.e. total number of state changed), in a synchronous model.
\gtodo{mention that move = number of state changes; mention the models: central/synchronous/distributed daemon}
In \cite{Ikeda2002}, the authors proposed a deterministic two-state algorithm that works under distributed scheduler (an adversary that, at each time, selects arbitrarily a set of processes to execute). Both algorithms stabilize in time $O(n^2)$.
In \cite{Turau07}, Turau introduces a 3-state self-stabilizing algorithm that stabilizes in $O(n)$ moves, under a distributed scheduler.
A breakthrough was achieved by Barenboim et al.\ \cite{BarenboimEG22}, who proposed a self-stabilizing algorithm for the MIS and other related problems, in the synchronous model. They prove that the algorithm stabilizes after $O(\Delta+ \log^* n)$ rounds. 

Assuming anonymous networks Shukla et al.\ \cite{Shukla1995} proposed two deterministic two-state self-stabilizing algorithms,  that work under a centralized scheduler (an adversary that selects one process to execute at each round) and
stabilizes in $O(n)$ rounds. 
In \cite{Turau19}, Turau introduced a synchronous randomized self-stabilizing algorithm for MIS that stabilizes w.h.p. in $O(\log n)$ rounds w.h.p. 
The possible states of the nodes are $O(\log n)$.



Next, we briefly summarize the best known upper  bounds to compute an MIS in the distributed LOCAL model on arbitrary graphs.
Barenboim et al. \cite{BarenboimEK14} proved that an MIS can be computed with a distributed deterministic algorithm in $O(\Delta + \log^*n)$ rounds and  Ghaffari et al.\ \cite{GhaffariGR21} provide an upper bound of $O(\log^5 n)$. Regarding distributed randomized algorithms, Ghaffari~\cite{Ghaffari16} provides an upper bound of $O(\log \Delta)+ 2^{O(\sqrt{\log \log n})}$ w.h.p., which, thanks to \cite{RozhonG20,GhaffariGR21}, was improved to $O(\log \Delta +\log^5 \log n)$ w.h.p. See also~\cite{FaourGGKR22}.


The current best-known lower bound for finding an MIS is proved by Balliu et al.\ \cite{BalliuBHORS21}, who show that computing an MIS in the LOCAL model requires $\Omega(\min\{\Delta, \log n/\log \log n\})$ rounds deterministically, and $\Omega(\min\{\Delta, \log \log n / \log \log \log n\})$ rounds with a randomized algorithm.



\small 
\bibliographystyle{plain}
\bibliography{misp}

\clearpage 

\end{document}